\documentclass[pra,aps,superscriptaddress,onecolumn,letterpaper,tightenlines,nofootinbib,11pt]{revtex4-2}
\usepackage{graphicx,color,amsmath,amsthm,amsfonts,enumerate,amssymb,mathtools,enumitem,thmtools,hyperref,mathdots,enumitem,centernot,bm,soul,bbm,stmaryrd}
\usepackage[capitalise, noabbrev]{cleveref}
\usepackage{multirow}
\usepackage{gnuplottex}
\usepackage[framemethod=TikZ]{mdframed}

\usepackage{tikz,pgfplots}
\usepackage{svg}
\pgfplotsset{compat=1.17}
\usetikzlibrary{quantikz, shapes, arrows, chains}
\hypersetup{colorlinks=true,linkcolor=blue,citecolor=blue,urlcolor=blue}

\usepackage{algorithm}
\usepackage{algpseudocode}

\usepackage{mathrsfs}

\usepackage{optidef}
\usepackage{thm-restate}

\let\originalleft\left
\let\originalright\right
\renewcommand{\left}{\mathopen{}\mathclose\bgroup\originalleft}
\renewcommand{\right}{\aftergroup\egroup\originalright}

\DeclareMathOperator*{\E}{\mathbb{E}}

\def\>{\rangle}
\def\<{\langle}
\newcommand{\abs}[1]{\left| {#1} \right|} 
\newcommand{\ketbra}[2]{\ensuremath{\left|#1\right\rangle\!\left\langle#2\right|}}

\newcommand{\tr}[1]{\mathrm{Tr}\left( #1 \right)}

\newcommand{\norm}[1]{\left\|#1\right\|}

\DeclareMathOperator*{\Var}{Var}
\DeclareMathOperator*{\Tr}{Tr}

\newcommand{\rhohat}{\hat{\rho}}

\newcommand{\mhat}{\hat{M}_1}

\newcommand{\phihat}{\hat{\phi}}

\definecolor{ppblue}{RGB}{46,117,182}
\definecolor{ppred}{RGB}{197, 90, 17}

\newcommand{\hphide}[1]{}

\newcommand{\inftynorm}[1]{\left\lVert #1 \right\rVert_{\infty}}

\newcommand{\pnorm}[2]{\left\lVert #1 \right\rVert_{#2}}

\newcommand{\bigo}[1]{\mathcal{O}\left(#1\right)}

\newcommand{\sym}{\Pi_{\operatorname{sym}}}

\newcommand{\symm}{\operatorname{S}}

\theoremstyle{plain}
\newtheorem{thm}{Theorem}
\newtheorem{theorem}[thm]{Theorem}

\newtheorem{lemma}[thm]{Lemma}
\newtheorem{prop}[thm]{Proposition}
\newtheorem{cor}[thm]{Corollary}

\newtheorem{fact}[thm]{Fact}
\theoremstyle{definition} 
\newtheorem{defn}[thm]{Definition}

\newcommand{\ve}{\mathbf{e}}
\newcommand{\vlambda}{\underline{\lambda}}

\newcommand{\vv}{\mathbf{v}}
\newcommand{\va}{\mathbf{a}}
\newcommand{\success}{\textup{success}}
\newcommand{\meas}{\textup{meas}}

\DeclareMathOperator*{\MP}{\mathrm{MP}}
\DeclareMathOperator*{\clone}{\mathrm{Cl}} % clone map

\newcommand{\CC}{\mathbb{C}}
\newcommand{\CCD}{\CC^{d}}
\newcommand{\fail}{\bot}

\DeclarePairedDelimiter\parens{\lparen}{\rparen}

\DeclarePairedDelimiter\bracks{\lbrack}{\rbrack}
\newcommand{\tracenorm}[1]{\norm{#1}_{\mathrm{tr}}}

\newcommand{\symdim}[2]{{#1}_{#2}} % old = {{#1}\bracks*{#2}}
\newcommand{\dens}[1]{\mathrm{Dens}\parens*{#1}} % e.g., Dens(C^d)
\newcommand{\lin}[1]{\mathcal{L}\parens*{#1}}

\newcommand{\eye}{\mathbb{I}}

\newcommand{\link}[1]{\lin{(\CCD)^{\otimes #1}}}
\newcommand{\densk}[1]{\dens{(\CCD)^{\otimes #1}}}

\newcommand{\truedist}{\mathscr{D}}
\newcommand{\fakedist}{\mathscr{D}'}

\newcommand{\samp}{\sim}

\newcommand{\initrho}{\rho}
\newcommand{\measrho}{\rho'}
\newcommand{\initeta}{\eta}
\newcommand{\measeta}{\eta'}

\newcommand{\romannum}[1]{%
\ifnum#1<1
\ifnum#1=0
o%
\else
-\romannumeral -#1%
\fi
\else
\romannumeral #1%
\fi}
\DeclareRobustCommand{\Romannum}[1]{\MakeUppercase{\romannum{#1}}}

\begin{document}
\title{Principal eigenstate classical shadows}
\author{Daniel Grier}
\affiliation{Department of Mathematics and Department of Computer Science and Engineering, University of California, San Diego}
\author{Hakop Pashayan}
\affiliation{Dahlem Center for Complex Quantum Systems, Freie Universit\"{a}t Berlin, Germany}
\author{Luke Schaeffer}
\affiliation{Joint Center for Quantum Information and Computer Science, University of Maryland, College Park}
% \date{}
\begin{abstract}
Given many copies of an unknown quantum state $\rho$, we consider the task of learning a classical description of its principal eigenstate. Namely, assuming that $\rho$ has an eigenstate $\ket{\phi}$ with (unknown) eigenvalue $\lambda > 1/2$, the goal is to learn a (classical shadows style) classical description of $\ket{\phi}$ which can later be used to estimate expectation values $\bra{\phi}O\ket{\phi}$ for any $O$ in some class of observables. 
We consider the sample-complexity setting in which generating a copy of $\rho$ is expensive, but joint measurements on many copies of the state are possible. We present a protocol for this task scaling with the principal eigenvalue $\lambda$ and show that it is optimal within a space of natural approaches, e.g., applying quantum state purification followed by a single-copy classical shadows scheme. Furthermore, when $\lambda$ is sufficiently close to $1$, the performance of our algorithm is optimal---matching the sample complexity for pure state classical shadows. 
\end{abstract}

\maketitle

\section{Introduction}

A key principle of algorithm design is to never do more work than is needed for the task at hand.
Consider the problem of identifying some unknown quantum state $\rho$ by measuring several copies of it. It has long been known that obtaining a complete description of such a state (say, by producing an estimate $\rhohat$ close in trace distance) requires a number of copies which grows linearly (or more) with the dimension of the Hilbert space. Such a strong requirement on the number of copies makes it nearly impossible to experimentally realize such tomographic protocols on all but the smallest quantum systems. 

Fortunately, a complete description of $\rho$ is unnecessary for many applications, allowing for dramatically simpler estimation protocols. Suppose, for example, you wish to estimate the fidelity of a state $\rho$ produced by an experimental quantum device with some target pure state, say, to benchmark your device. In this case, the number of copies you must prepare scales only with your desired precision, not with the dimension of the ambient space, making the entire procedure much more practical.

The fidelity estimation protocol is a special case of a recent and enormously popular framework for predicting properties of unknown quantum states called \emph{classical shadows} introduced by \cite{Huang2020}. In this setting, many copies of the unknown state $\rho$ are measured and condensed into a classical bit string. This classical description can later be used to estimate $\Tr(O \rho)$ for any $O$ in some class of Hermitian observables with very high probability. The success of the classical shadows framework \citep{cerezo2021variational, bharti2022noisy} motivates a deeper consideration into how it can be further improved to model practical quantum learning tasks as well as how it can be made more sample efficient.

In many practical scenarios it is not properties of the state $\rho$ that one wants to learn, but rather those of its top eigenstate. A natural setting where principal eigenstates become the focal object is when one only has access to noisy copies of a target state $\ket{\phi}$. In the case of global depolarizing noise acting on a $d$-dimensional Hilbert space, the noisy state is $\rho = (1-\eta) \proj{\phi} + \eta (I-\proj{\phi})/(d-1)$. For $\eta<1/2$, the principal eigenstate remains $\ket\phi$. Hence, in the case of global depolarizing noise, the target state can be recovered from the principal eigenstate. 
In fact, this remains true for other practical noise models as well \citep{koczor2021dominant, huggins2021virtual}. Furthermore, settings where one only has access to noisy copies are natural. For example, consider a scenario where copies of a noisy quantum state $\rho$ are prepared by a quantum sensor operating in non-ideal environmental conditions and fed into a powerful quantum processor to extract data. Indeed Ref.~\cite{Yamamoto2022EMQMetrology} considers such a scenario in the setting of unknown fluctuating noise. Other proposed applications in related work include noise suppression for noisy intermediate-scale quantum computation~\cite{koczor2021exponential, huggins2021virtual, czarnik2021qubit, Seif2023ErrorMitigation, zhou2022hybrid}.

In this paper, we ask what happens when you combine classical shadows with principal eigenstate estimation. Namely, what is the complexity of estimating observable expectation values with respect to the dominant eigenvector of $\rho$ rather than $\rho$ itself? To this end, we introduce the following ``principal eigenstate classical shadows'' task:

\begin{mdframed}[align=center, userdefinedwidth=.98\linewidth, frametitle={Principal eigenstate classical shadows}, frametitlerule=true, frametitlebackgroundcolor=gray!15!, roundcorner=5pt, skipbelow=-500pt, skipabove=10pt]
\begin{tabular}{l | c l c l }
\multirow{2}{*}{Learning} & \hspace{2pt} &\textit{Input:} && Copies of $\rho = (1-\eta) \phi + \eta \sigma$ with principal \\ &&&& eigenstate $\phi = \proj{\phi}$, $\eta < 1/2$, and $\Tr(\phi\sigma) = 0$ \hspace{-5pt} \\
 && \textit{Output:} && Classical description $\hat\phi$ \\  [.5em] \hline 
 \\ [-1em]
\multirow{2}{*}{Estimation} \hspace{5pt} & &\textit{Input:} && Hermitian observable $O$ with $\lVert O \rVert_\infty \leq 1$ and classical description $\hat\phi$ \\
&& \textit{Output:} && Compute $E$ such that $| \bra\phi O \ket\phi - E| \le \epsilon$
\end{tabular}
\end{mdframed}

In this work, we focus on the goal of solving the principal eigenstate classical shadows problem with the fewest copies of the input state $\rho$. That is, we want to determine the \emph{sample complexity} of this task since producing a copy of $\rho$ is usually considered to be a resource-intensive task.

One of the key parameters of this task is $\eta$---the \emph{principal deviation}---which determines how far $\rho$ is from its principal eigenstate. Notice that it is this deviation that prevents traditional classical shadow approaches from achieving a high degree of accuracy on this task. That is, classical shadows protocols for the state $\rho$ can only be accurate up to additive error $\eta$ on the state $\ket\phi$ since $|\mathrm{Tr}(O\rho) - \bra\phi O \ket\phi| = \eta$ for $O = \proj\phi$.

Nevertheless, even if you were satisfied with an estimate to accuracy $\eta$---a setting in which you could theoretically still use traditional classical shadows approaches---there is still reason to suspect that more sample-efficient algorithms exist. Intuitively, traditional shadow estimation algorithms do not take advantage of the purity of the underlying state $\ket{\phi}$ that we wish to measure. Indeed, in the noiseless setting (i.e., $\eta = 0$), any shadow algorithm which does not take advantage of the purity of the underlying state is provably suboptimal \citep{gps2022sampleoptimal}. Furthermore, the sample complexity may improve by using joint measurements---i.e., measurements on multiple copies of $\rho$ simultaneously (i.e., $\rho^{\otimes n}$ for some $n$) using a single entangling positive operator-valued measure (POVM).

We will show that this intuition is correct. That is, we give algorithms for the principal eigenstate classical shadows problem that leverage these insights to be significantly more sample-efficient than classical shadow protocols for generic states.
One might quite naturally wonder if our algorithm is simply the combination of other powerful subroutines for processing quantum states, of which many are possibly relevant \citep{cirac1999optimal, keyl2001rate, HonghaoFu2016thesis, childs2023streaming, ODonnel2016efficient}. For example, it is true that we could design an algorithm for principal eigenstate classical shadows by first de-noising (sometimes referred to as ``purifying'') the state $\rho$ into the state $\ket\phi$ \citep{cirac1999optimal, keyl2001rate, lloyd2014PCA, HonghaoFu2016thesis, childs2023streaming}, and then applying a classical shadows algorithm \citep{Huang2020}. We show that the sample complexity of this approach is \emph{worse} than the robust algorithm of this paper which solves the principal eigenstate shadows problem directly. In other words, even though both of the subroutines mentioned above are optimal for their respective subtasks, they are nevertheless doing more work than needed when combined to perform observable estimation on the principal eigenstate.

\subsection{Main result}    
Our classical shadows protocol is actually a suite of algorithms that depend on the principal deviation $\eta$ of the underlying state $\rho = (1-\eta) \phi + \eta \sigma$. We do not require a priori knowledge of $\eta$. 
We will see that there are three $\eta$ regimes and as $\eta$ decreases from $1/2$ to zero, sample complexity also decreases reaching a minimum at the optimal sample complexity for learning pure states~\citep{gps2022sampleoptimal}. Surprisingly, this occurs before $\eta$ becomes zero. These sample complexities are given in the following theorem:

\begin{theorem} 
\label{thm:sample_complexity_main}
There exists a protocol (comprised of separate learning and estimation algorithms) for solving the principal eigenstate classical shadows task with high probability that has three regimes of sample complexity determined by the deviation $\eta$ shown below
\begin{center}
\begin{tikzpicture}[font = \small, line/.style = {draw,thick, shorten >=-2pt, shorten <=-2pt}]
\coordinate (eta1) at (4,0);
\coordinate (eta2) at (8,0);
\coordinate (eta3) at (12,0);
\coordinate (end) at (13,0);

\draw[->] (0,0) -- (end);
\draw (0, 0.15) -- (0,-.15);
\draw ($ (eta1) + (0,0.15) $) -- ($ (eta1) - (0,0.15) $) ;
\draw ($ (eta2) + (0,0.15) $) -- ($ (eta2) - (0,0.15) $) ;
\draw ($ (eta3) + (0,0.15) $) -- ($ (eta3) - (0,0.15) $) ;

\node[yshift = 8] at ($ (0,0)!0.5!(eta1) $) {$\mathcal O(s^*)$};
\node[yshift = 8] at ($ (eta1)!0.5!(eta2) $) {$\mathcal O(\frac{B\eta + 1}{\epsilon^2})$};
\node[yshift = 8] at ($ (eta2)!0.5!(eta3) $) {$\mathcal O ( \frac{B\eta}{\epsilon^2} + \frac{\eta}{\epsilon^{5/2}} )$};

\node[yshift = -12] at (0,0) {$0$};
\node[yshift = -12] at (eta1) {$1/s^*$};
\node[yshift = -12] at (eta2) {$\sqrt{\epsilon}$};
\node[yshift = -12] at (eta3) {$1/2$};
\node[xshift = 5] at (end) {$\eta$};

\end{tikzpicture}
\end{center}
where $B \geq \Tr(O^2)$ is the squared-Frobenius norm of observable $O$ and $s^* := \frac{\sqrt{B}}{\epsilon} + \frac{1}{\epsilon^2}$ is the optimal sample complexity for classical shadows on pure states.\footnote{Technically, solely in the $\eta \in (\sqrt{\epsilon}, 1/2)$ regime, we invoke a purification procedure of \cite{HonghaoFu2016thesis, childs2023streaming} that only works on depolarized states, i.e., $\rho = (1-\eta) \phi + \eta \sigma$ for $\sigma = (I - \phi)/(d-1)$. However, based on ongoing/unpublished work, we claim that this purification procedure can be generalized to allow for arbitrary $\sigma$.} Furthermore, in all regimes, the protocol incurs a $\bigo{\log M}$ factor in sample complexity to approximate $M$ observable expectation values (all to $\epsilon$ accuracy) with high probability.
\end{theorem}
    
These bounds may at first seem somewhat arbitrary, so let's spend a few moments to put them in context. First, it is worth noting that the optimal sample complexity in the zero-deviation setting is $\tilde\Theta ( \epsilon^{-1} \sqrt B + \epsilon^{-2})$ as determined by \cite{gps2022sampleoptimal}. In other words, in the first regime where the deviation is quite small (i.e., $\eta \leq 1/s^*$), the sample complexity is identical to that of the optimal zero-deviation measurement protocol. Moreover, our measurement protocol in this regime is actually the same measurement procedure as used for zero deviation. However, this is not to say that the analysis of this protocol is trivial or in any sense a black box reduction to the pure case. In fact, the bulk of the technical work in this paper is spent addressing this setting.
    
The sample complexities in the remaining two settings are shown in some sense by a reduction to the first setting. In the second regime, we measure independent blocks of the unknown state, and post-process these measurement results with a simple averaging procedure. In the third regime, when the noise is the most extreme, we must first pre-process the input by an explicit purification protocol that uses multiple copies of $\rho$ to distill a new quantum state with smaller deviation. We then proceed by invoking the procedure above. For this procedure, we determine the optimal choice of the number of copies to be jointly measured, averaged, and purified. We also present an estimation protocol for $\eta$ that can be used to achieve an overall sample complexity which, up to big-$\mathcal{O}$ notation, matches that of our procedure when using an optimal choice of parameters and a \emph{known} $\eta$.
    
Finally, we note that we can always employ the standard median-of-means trick~\citep{lugosi2019mean,lerasle2019lecture,Huang2020} to amplify the success probability of obtaining an accurate estimate. In this way, to estimate $M$ distinct observables $\{O_i \mid \Tr(O_i^2) \le B \}_{i=1}^M$, we incur a factor of $\log(M)$ in the sample complexity, as is typical with shadow estimation protocols \citep{Huang2020}.
    
\subsection{Technical challenges}

One of the central tools used in tomographic protocols for pure states \citep{massar1995optimal, PureStateTomography, gps2022sampleoptimal} is a continuous POVM proportional to $\{ \proj{\psi}^{\otimes n} \}_{\psi}$ that we call the \emph{standard symmetric joint measurement} (see \Cref{defn:abbreviated_measurement}). 
Intuitively, this POVM is more likely to output a measurement outcome $\psi$ the closer $\ket{\psi}$ is to the measured unknown state $\rho$. Not only is this measurement optimal for pure state tomography \citep{massar1995optimal}, but it is also surprisingly easy to analyze in many cases due to its tight connection with representation theory and the symmetric subspace \citep{harrow2013church}. Indeed, it is this connection that allows for a simple analysis of the original classical protocol \citep{Huang2020, mele2023introduction}.

The main conceptual bottleneck for our analysis is that the unknown state $\rho$ may have small overlap with the symmetric subspace, rendering the standard symmetric joint measurement ineffective. Of course, there are more powerful representation-theoretic tools for learning mixed states, but these tend to incur a factor of the Hilbert space dimension, which is prohibitively large for many applications of classical shadows \citep{Haah2016sample, ODonnel2016efficient}. That said, when the deviation $\eta$ is small, $\rho$ will still be close enough to the symmetric subspace that the standard symmetric joint measurement will succeed. The major technical contribution of this paper is showing that this measurement also serves as a sample-efficient predictor for the principal eigenstate when used in the context of classical shadows.

\begin{restatable}{theorem}{robustnessjoint}
\label{thm:robustness_of_joint_measurement}
    Let $\rho=(1-\eta)\phi+\eta \sigma$ be an instance of principal eigenstate classical shadows. The standard joint measurement on $n$ copies of state $\rho$ succeeds with probability at least $(1-\eta)^{n-1}$. Conditioned on success of the measurement, there is an estimator $\hat\phi$ such that
    \begin{align*}
        \E[\Tr(O\hat\phi)] &= \Tr(O\phi) + \frac{\eta \Tr(O(\sigma - \phi))}{n(1 - \eta)} + \bigo{\eta^2/n} \\
        \Var[\Tr(O\hat\phi)] &= \frac{\Tr(O^2)}{n^2} + \bigo{\eta^2 + 1/n}
    \end{align*}
    for observables $O$ with $\inftynorm{O} \le 1$.
\end{restatable}

Notice that given \Cref{thm:robustness_of_joint_measurement}, one can easily derive the sample complexity of the $\eta \le 1/s^*$ regime given in \Cref{thm:sample_complexity_main} by invoking Chebyshev's inequality. Recall that the next two regimes are obtained by averaging, purification, or a combination thereof. We describe those procedures in \Cref{sec:everything_together}. We give a detailed proof outline for \Cref{thm:robustness_of_joint_measurement} in \Cref{sec:proof_outline}, with full proofs in the appendix.

\subsection{Related Work}
The task of approximately preparing the principal eigenstate (a.k.a. purification) has a long history~\citep{Berthiaume1994, Barenco1997, Peres1999, werner1998optimal, cirac1999optimal, lloyd2014PCA, HonghaoFu2016thesis, childs2023streaming}. However, this task is costly achieving an $\eta$ suppression that scales at most inversely in the number of copies of $\rho$~\citep{werner1998optimal, cirac1999optimal}. This task is distinct from our work which aims to learn a classical description of the principal eigenstate to sufficient accuracy to permit future estimates of many expectation values. 

More recently, focus has shifted from physical to virtual purification schemes. This more relevant body of previous work~\citep{Ekert2002, colter2019cooling, koczor2021exponential, huggins2021virtual, czarnik2021qubit}, sometimes referred to as ``virtual distillation'', is one that directly learns an expectation value of a given observable with respect to $\rho_t :=\rho^t/\Tr(\rho^t)$ for an unknown state $\rho$ and an integer $t$. As $t$ becomes large, $\rho_t$ approaches the principal eigenstate of $\rho$ connecting these techniques to our work. 
However, in contrast to our work, in these protocols, the observable is a part of the measurement circuit. So, for example, computing expected values for exponentially many observables would require exponential overhead in sample complexity, whereas our procedure in \Cref{thm:sample_complexity_main} requires only polynomial overhead. To be fair, there are certainly advantages to the virtual distillation setting for practical applications, most notably the fact that the sample complexity does not depend on properties of the observable such as its Frobenius norm.

Building on the virtual distillation program, Refs.~\cite{zhou2022hybrid, Seif2023ErrorMitigation} consider learning non-linear functions of $\rho$ such as $\Tr(O\rho^t)$; however, like our work, Refs.~\cite{zhou2022hybrid, Seif2023ErrorMitigation} take a classical shadows style approach where measurements of the copies of $\rho$ can be implemented without knowing the observables of interest. Ref.~\cite{zhou2022hybrid} shows that their sample complexities depend on $B$, the squared-Frobenius norm of the observable, achieving a sample complexity of $\bigo{\frac{(B+1)t}{\epsilon^2}}$ for the sub-procedure of estimating $\Tr(O\rho^t)$ to additive error $\epsilon$. 
A straightforward calculation shows that solving the principal eigenstate classical shadows problem using these techniques for the estimator $\Tr(O\rho^t)/\Tr(\rho^t)$ results in much higher sample complexities for all regimes in which $B$ is somewhat large (in particular, when $B > \eta/\sqrt{\epsilon}$). Consequently, our protocol is better for all regimes which do not rely on purification (i.e., $\eta \le \sqrt{\epsilon}$). A more involved calculation shows that our protocol is still preferable in all but a handful of regimes, but they are harder to characterize (e.g., $B = 1$, $\epsilon = \eta^3$, and $\eta$ sufficiently small). We leave a more thorough comparison of these techniques (including possible ways to combine them) to future work. 
We note that Ref.~\cite{zhou2022hybrid} also considers the setting where $O$ is a $k$-local observable. There, the sample complexity of estimating $\Tr(O \rho^{t})$ is $\bigo{\frac{4^k t}{\epsilon^2}}$. Hence, for $k$-local observables where $B=4^n \gg 4^k$, this protocol is preferable in several parameter regimes of interest.

Another related body of work surrounds classical shadows that are robust to noise in the measurement process itself \citep{chen2021robust, koh2022classical}. In other words, those procedures work well when given a state that has been prepared with high fidelity, but are using low-fidelity measurements. Our procedure works well when given a low-fidelity state, but have measurements that can be performed with high fidelity.

\subsection{Open Problems}
Our work leaves open several new directions. Perhaps the most interesting is to explore variations of the principal eigenstate classical shadows problem.  How do shadow estimation algorithms need to change to predict properties of the principal eigenstate, rather than the state itself? There are many possible variants worthy of consideration: when the class of observables is local \citep[cf.][]{Huang2020, hakoshima2023localizedVP}; when the measurement procedure itself is faulty \citep[cf.][]{chen2021robust, koh2022classical}; when a low memory footprint is required \citep[cf.][]{czarnik2021qubit, chen2022exponential, hakoshima2023localizedVP}; etc.  

Another possible direction for future work is to generalize the principal eigenstate classical shadows problem to the top $k$ eigenstates, rather than just the top eigenstate. When a complete description of the best rank-$k$ approximator is needed,  $O(k d / \epsilon^2)$ samples are sufficient by work of \cite{ODonnel2016efficient} (where $d$ is the dimension of the Hilbert space and $\epsilon$ is a bound on the trace distance to the optimal rank-$k$ approximation), but once again, little is known in the classical shadows setting.

Finally, we ask whether or not our algorithm can be improved. When $\eta \le 1/s^*$ (recall $s^* := \sqrt B/\epsilon + 1/\epsilon^2$), our algorithm obtains the same sample complexity as the $\eta = 0$ algorithm of \cite{gps2022sampleoptimal}, which is provably optimal up to log factors. Therefore, our algorithm must also be optimal in that regime since we could always add noise in $\eta = 0$ setting if that improved the sample complexity. In the regime where $\eta > 1/s^*$, the optimality of our algorithm is unknown. However, one might suspect that there is the possibility for improvement since our algorithm does not measure all (or almost all) copies of $\rho$ at once, which is distinct from other optimal joint-measurement tomography algorithms \citep{Haah2016sample, ODonnel2016efficient, gps2022sampleoptimal}.

\section{Abbreviated Preliminaries}
\label{sec:condensed_preliminaries}
We start with a condensed version of the preliminaries section in the appendix (cf.\ \Cref{sec:preliminaries}). Let $\symm_{n}$ denote the symmetric group of permutations on $n$ elements.

\begin{defn}[permutation operator]
    Given a permutation $\pi \in \symm_n$ (for $n \geq 1$), define a permutation operator $W_{\pi} \in \CC^{d^{n} \times d^{n}}$ such that 
    $
    W_{\pi} \ket{x_1} \cdots \ket{x_n} = \ket{x_{\pi^{-1}(1)}} \cdots \ket{x_{\pi^{-1}(n)}},
    $
    and extend by linearity. 
    That is, $W_{\pi}$ acts on $(\CCD)^{\otimes n}$ by permuting the qudits, sending the qudit in position $i$ to position $\pi(i)$. 
\end{defn}

\begin{defn}[symmetric subspace]
    Define the \emph{symmetric subspace} as the subspace of $(\mathbb C^d)^{\otimes n}$ fixed by the projector $\sym^{(n)} = \frac{1}{n!} \sum_{\pi \in \symm_n} W_{\pi}$, where 
    $\symdim{d}{n} = \binom{n+d-1}{d-1}$. Additionally, we have that $
    \sym^{(n)} = \symdim{d}{n} \int_{\psi} \left( \ketbra{\psi}{\psi} \right)^{\otimes n} \mathrm{d}\psi
    $ \citep[e.g.,][]{scott2006tight}.
\end{defn}

\begin{defn}
\label{defn:abbreviated_measurement}
The \emph{standard symmetric joint measurement} is a measurement on $n$ qudits. It is defined by the POVM $\mathcal{M}_{n} = \{ F_{\psi} \}_{\psi} \cup \{ F_{\fail} \}$ with elements
$
F_{\psi} := \symdim{d}{n} \cdot \ketbra{\psi}{\psi}^{\otimes n} \mathrm{d} \psi,
$
for all $d$-dimensional pure states $\psi$, proportional to the Haar measure $\mathrm{d} \psi$, plus a ``failure'' outcome $F_{\fail} := \eye - \sym^{(n)}$ for non-symmetric states.
\end{defn}

The measurement \emph{fails} if we get outcome $F_{\fail}$, otherwise we say it \emph{succeeds}. In what follows, we will let $\Psi$ be the random variable representing the outcome of measuring $\rho^{\otimes n}$ with $\mathcal{M}_n$ where $
\Psi$ is $0$ when the measurement fails and $\proj{\psi}$ for measurement outcome $\psi$.

Note that standard techniques (employing $t$-designs, see \cite{bajnok1992construction, hayashi2005reexamination, bondarenko2013optimal}) can be used to replace this measurement with a POVM with finitely many outcomes. Subsequently, the finite POVM can be compiled to a projective measurement \cite{nielsen2010quantum}. See Ref.~\cite{gps2022sampleoptimal} for more details on how to realize this measurement. All our sample complexity bounds hold under this replacement.

\section{Outline of main theorem}
\label{sec:proof_outline}

We now give an outline of the proof of \Cref{thm:robustness_of_joint_measurement} to elucidate some of the key techniques. We refer the reader to the appendix for full proofs and details. 

Let's begin with the most straightforward approach to proving \Cref{thm:robustness_of_joint_measurement}---simply give exact expressions for the first and second moments of the standard symmetric measurement on $\rho^{\otimes n}$ conditioned on a successful\footnote{As mentioned before, the fact that our measurement can fail (and output $0$) is the consequence of our state not necessarily being in the symmetric subspace.} outcome:
\begin{theorem}[\Cref{thm:moments_of_psi_given_ms} in \Cref{sec:chiribella}]\label{thm:short_psi_moments}
    \begin{align}
        \E[ \Psi \mid \success] &= \frac{\eye + n M_1}{d + n}, \\
        \E[ \Psi^{\otimes 2} \mid \success] &= \frac{2 \sym^{(2)}}{(d+n)(d+n+1)} \parens*{(\eye + nM_1)^{\otimes 2} + \binom{n}{2} M_2 - n^2 M_1^{\otimes 2}},
    \end{align}
for mixed states $M_1 \propto \Tr_{1,\ldots,n-1}(\sym^{(n)} \rho^{\otimes n})$ and $M_2 \propto \Tr_{1,\ldots,n-2}(\sym^{(n)} \rho^{\otimes n})$.

\end{theorem}

Here, we are already forced to deviate from previous treatments \citep{Huang2020, gps2022sampleoptimal}. Notice that the expected value is not related by scaling and shifting by the identity to the unknown state $\rho$. Instead, the measurement's expectation is related to $M_1$, the partial trace of the projection of $\rho^{\otimes n}$ onto the symmetric subspace. As in \citep{gps2022sampleoptimal}, the proof of \Cref{thm:short_psi_moments} relies on representation theory and properties of the symmetric subspace, but is considerably more streamlined by the use of \nameref{thm:chiribella}.

Ultimately, we will claim that $M_1$ is close to the principal eigenstate $\phi$ of $\rho$. So, the estimator $\hat\phi$ in \Cref{thm:robustness_of_joint_measurement} will be $((d+n)\Psi  - \mathbb I)/n$ conditioned on successful measurement. Two key challenges remain: first, we must show that $M_1$ is actually close to the principal eigenstate; second, we must bound the variance of our estimator. Unfortunately, the closed-form expressions for $M_1$ and $M_2$ are quite unwieldy. 

To tackle these challenges, we reinterpret $\rho^{\otimes n}$ as a statistical mixture of states which are easier to analyze individually. To describe this decomposition, first let us write the unknown state as $\rho = \sum_{i=1}^{d} \lambda_i \Phi_i$ where $\lambda_1 \ge \cdots \ge \lambda_d$ and $\Phi_i := \ketbra{\phi_i}{\phi_i}$ are projectors onto the eigenstates. In the expansion of $\rho^{\otimes n}$, we will use vectors $\ve = (e_1, \dotsc, e_d) \in \mathbb N^{d}$ with $e_1 + \dotsm + e_d = n$ to give counts for the different eigenstates of $\rho$. Now, we can define the mixed state
\[
\sigma(\ve) := \frac{1}{n!} \sum_{\pi \in \symm_n} W_{\pi} \left(\Phi_1^{\otimes e_1} \otimes \cdots \otimes \Phi_d^{\otimes e_d} \right) W_{\pi}^{\dag}
\]
to be a symmetrized\footnote{To be clear, $\sigma(\ve)$ is typically \emph{not} in the symmetric subspace since it is in general a mixed state.} version of the eigenstate $\Phi_1^{\otimes e_1} \otimes \cdots \otimes \Phi_d^{\otimes e_d}$. Using the short-hand expressions: $\ve! := e_1! \dotsm e_d!$, $\binom{n}{\ve} := \frac{n!}{\ve!}$, and $\vlambda^{\ve} := \lambda_1^{e_1} \dotsm \lambda_d^{e_d}$, we obtain our desired nice expansion of $\rho^{\otimes n}$:

\begin{prop} $\rho^{\otimes n} = \sum_{\ve} \binom{n}{\ve} \vlambda^{\ve} \sigma(\ve)$.\end{prop}

We now interpret $\rho^{\otimes n}$ as statistical mixture of $\sigma(\ve)$ states where $\ve$ is selected at random from the distribution with $\Pr[\ve] = \binom{n}{\ve} \vlambda^{\ve}$, which we recognize as the multinomial distribution with $n$ trials for $d$ events with probabilities $\lambda_1, \ldots, \lambda_d$. However, the pertinent distribution for our calculations, which we name $\truedist$, is this multinomial \emph{conditioned} on successful measurement. We arrive at new expressions (c.f.\ \Cref{thm:ms_from_mse}) for the first and second moments in \Cref{thm:short_psi_moments} by expanding $\rho^{\otimes n}$ with this interpretation. For example, the first moment becomes
\[
    \E[ \Psi \mid \success] = \frac{\eye + n \mathbb E_{\ve \sim \truedist} M_1(\ve)}{d + n}
\]
where mixed state $M_1(\ve) \propto \Tr_{1,\ldots,n-1}(\sym^{(n)} \sigma(\ve))$. This expansion has the potential to greatly simplify the calculation since $M_1(\ve)$ and $M_2(\ve)$ turn out to have surprisingly clean forms (c.f.\ theorems \ref{thm:sigma1} and \ref{thm:sigma2}, respectively). Unfortunately, the distribution $\truedist$ is still quite complicated.

To circumvent this issue, our key observation is that the true distribution $\truedist$ is close to a distribution $\fakedist$ of independent geometric random variables. Technically, in $\fakedist$, each $e_i$ with $i \ge 2$ is chosen independently from the geometric distribution with mean $\lambda_i / (\lambda_1 - \lambda_i)$ and $e_1$ is set to $n - (e_2 + \ldots + e_d)$.
\begin{theorem}[\Cref{thm:prob_e_negative} in \Cref{sec:geometric_approximation}]
$\norm{\truedist - \fakedist}_{TV} \le \parens*{\frac{1 - \lambda_1}{\lambda_1}}^{n+1} \frac{\lambda_1}{2\lambda_1 - 1}$.
\end{theorem}
In other words, we can substitute $\truedist$ for $\fakedist$ without significant loss.\footnote{This idea is similar, but not identical to a technique called ``poissonization'' \citep{dasgupta2011probability}.} This geometric approximation dramatically simplifies many calculations, but nevertheless requires care to show it does not significantly affect the variance of our estimator, involving a sort of hybrid calculation where sometimes we assume the approximation and sometimes we do not. We leave these details to \Cref{sec:mean_of_estimator} and \Cref{sec:variance_of_estimator} for the first and second moments, respectively. Combining these pieces together completes the proof.

\section{The compound estimation procedure}
\label{sec:everything_together}

The full estimation procedure (to prove Theorem~\ref{thm:sample_complexity_main}) uses our \emph{measurement} in Theorem~\ref{thm:robustness_of_joint_measurement} as a black box, which is combined with \emph{purification} of $\initrho$ before measurement, and \emph{averaging} estimates from multiple measurements. We also require a step to estimate $\eta$ from samples, to decide the $\eta$-regime of Theorem~\ref{thm:sample_complexity_main} and balance the \emph{purification}, \emph{measurement}, and \emph{averaging} subroutines accordingly. Due to randomness in these subroutines, we will bound the expected number of samples. 

    In the purification step, we assume the existence of a black box which takes copies of $\initrho$ and creates a state $\measrho$ as output. The number of copies consumed in this sub-procedure is a random variable with mean $k$ but the output state is deterministic in the sense that for identical inputs $\initrho$, identical outputs $\measrho$ will be produced independent of the number of copies consumed. The purification procedure reduces the principal deviation, i.e., the deviation of $\measrho$ satisfies:
    \begin{align}
    \measeta=\bigo{\initeta/k},\label{eq:purified_deviation}    
    \end{align}
    where $\initeta<1/2$ is the deviation of $\initrho$. This result was shown to hold in the special case of $d=2$~\citep{werner1998optimal, cirac1999optimal}. This result was later shown to hold in the general $d$ setting in the special case where $\initrho$ is a convex combination of a pure quantum state and the maximally mixed state~\citep{HonghaoFu2016thesis, childs2023streaming}. Based on  unpublished work, we claim that this result holds in greater generality: it applies to arbitrary mixed states in arbitrary dimension subject to $\eta<1/2$. In our estimation procedure, we employ this result in the $\initeta \in (\sqrt{\epsilon},1/2)$ regime.
    
    In the measurement step, $n$ copies of $\measrho$ are consumed and an estimator $\hat\phi$ is output. This computation involves two steps. First, $n$ copies of $\measrho$ are measured using the standard symmetric joint measurement (cf.~\cref{defn:abbreviated_measurement}) producing either a fail outcome or a classical description of a pure state $\Psi$. If a fail outcome is observed, the execution of the measurement sub-procedure fails on this instance resulting in $n$ ``wasted'' copies of $\measrho$. We will be interested in the regime where the measurement sub-procedure will be executed many times with each having a constant probability of success hence, failures will at most contribute a constant factor to sample complexity. If the measurement succeeds, the measurement outcome is a classical description of a pure quantum state $\Psi$. An affine map is applied to produce $\hat\phi$, an estimator for $M_1$ and $\Phi_1$:
    \begin{align}
        \hat\phi=\frac{(d+n)\Psi-I}{n}.
    \end{align}
    This process is probabilistic so each call produces a different $\hat\phi$ with mean $M_1$ and variance given by \cref{cor:estimator_mean_and_variance_wrt_ms}.    
    
    In the averaging step, $b$ independent estimates of $M_1$ are averaged to produce one improved estimate $\hat\phi(b)$. This has mean $M_1$ and a variance $\frac{1}{b}$ times that of $\hat\phi$. The estimator $\hat\phi(b)$ is an unbiased estimator of $M_1$ and a biased estimator of the principal eigenstate $\Phi_1$. Using \cref{thm:robustness_of_joint_measurement}, for observables satisfying $\pnorm{O}{\infty}\leq 1$ the bias can be bounded by: %(cf. \cref{eq:bias_UB}):
    
    \begin{align}
        \beta=\abs{\E[\Tr(O(\hat\phi-\phi))]}
        = \bigo{\frac{\measeta}{n}}.\label{eq:bias_norm_operator}
    \end{align}
    By ensuring that our estimator has bias $\mathcal O(\epsilon)$ and variance $\mathcal O(\epsilon^2)$, we employ Chebyshev's inequality to prove \cref{thm:sample_complexity_main}.

%--------------------------------------------------------------
%------------- FIG:Combined PURIFY-MEASURE-AVERAGE --------------------
%--------------------------------------------------------------
\begin{figure}[ht]
    \centering
    \scalebox{.8}{\includegraphics{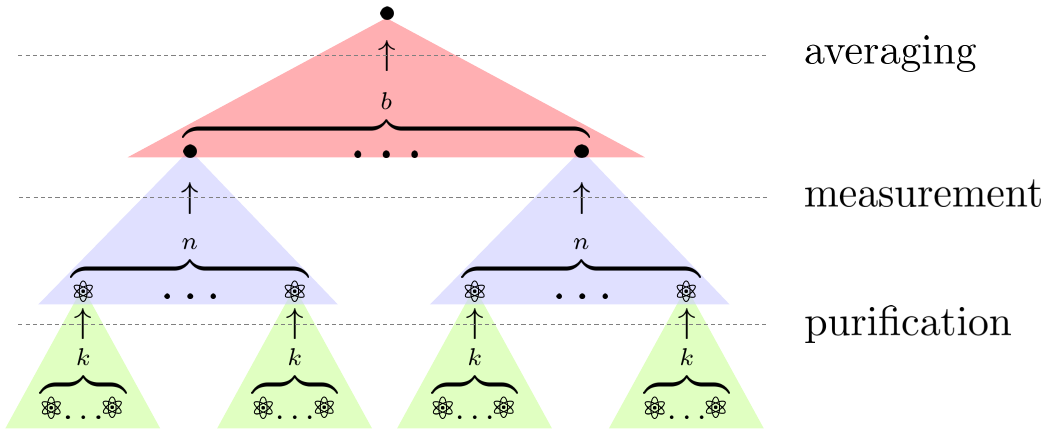}}
    \caption{Our three step estimation procedure depicting the purification, measurement and averaging sub-procedures (from bottom to top). The purification procedure maps $k$ quantum states to one quantum state (depicted by atom logos). The measurement procedure maps $n$ quantum states to a classical description of an operator (depicted by the ``$\bullet$'' symbol). The averaging procedure maps $b$ classical descriptions to one classical description of an operator.}
    \label{fig:combined_subprocedures}
\end{figure}
%--------------------------------------------------------------
%--------------------------------------------------------------
%--------------------------------------------------------------
    
    \Cref{fig:combined_subprocedures} shows how these sub-procedures are combined to form our estimator $\hat\phi(b)$. For a given observable $O$, an expectation value  can be estimated using $\Tr(O \hat\phi(b))$. With constant probability of failure (over the randomness of the measurement procedure), this produces an estimate of $\tr{O\phi}$ up to additive error $\epsilon$. By repeating this procedure and taking the median value over all repetitions, the probability of failure can be exponentially suppressed in the number of repetitions~\citep{lugosi2019mean, lerasle2019lecture, Huang2020}. We omit this standard ``median-of-means'' sub-procedure from our analysis, but note that it ensures that $\bigo{\log t}$ repetitions suffice to estimate the expectation value of $t$ observables, $O_1, \ldots, O_t$, all to within additive error $\epsilon$ with high probability.

The remainder of this section will discuss the choice of parameters $k, n$ and $b$ and how the performance of our procedure compares to alternative approaches.
    
    \subsection{Choice of parameters}

    Three parameters ($k$, $n$, and $b$) define the algorithm, and control both the accuracy of our estimate and the expected number of samples of $\initrho$ used. We select values for these parameters based on the given values of $B$ and $\epsilon$, as well as $\eta$. Note that $\eta$ is \emph{not} given, but let us suppose we know it for now and come back to the problem of estimating $\eta$ from samples after the theorem. %at the end of this section.

\begin{restatable}{theorem}{optimalchoice}
    \label{thm:optimal_parameter_choice}
    Given $B$, $\epsilon$ and $\eta$, 
    the expected number of samples is minimized for the choice of $k$, $n$ and $b$ given in \cref{table:choice}.
\end{restatable} 

The expectation is over the randomness in the purification and measurement procedures. The proof of Theorem~\ref{thm:optimal_parameter_choice} is given in Appendix~\ref{sec:optimal_parameter_choice}.

\begin{table}[h]
    \centering
    \scalebox{1}{
    \begin{tabular}{c|c|c|c}
    $\eta$ & $\bigo{1/s^{*}}$ & $\Omega(1/s^{*}) \cap \bigo{\sqrt{\epsilon}}$ & $\Omega(\sqrt{\epsilon})$ \\
    \hline
    $k$ & $1$ & $1$ & $\bigo{\frac{\initeta}{\sqrt{\epsilon}}}$ \\
    $n$ & $\bigo{s^*}$ & $\bigo{\frac{1}{\initeta}}$ & $\bigo{\frac{1}{\sqrt{\epsilon}}}$ \\
    $b$ & $1$ & $\bigo{\frac{B\initeta^2+\initeta}{\epsilon^2}}$ & $\bigo{\frac{B}{\epsilon}+\frac{1}{\epsilon^{3/2}}}$ \\
    \hline
    $s$ & $\bigo{s^*}$ & $\bigo{\frac{B \eta + 1}{\epsilon^2}}$ & $\bigo{\frac{B\eta}{\epsilon^2} + \frac{\eta}{\epsilon^{5/2}}}$
    \end{tabular}
    }
    \caption{Choice of parameters $k$, $n$, $b$ for the three regimes of $\eta$. Recall $s^* := \frac{\sqrt{B}}{\epsilon} + \frac{1}{\epsilon^2}$.}
    \label{table:choice}
\end{table}

Last, we need a way to estimate $\eta$, since the choice of $k$, $n$, $b$ are functions of $\eta$, either explicitly, or because they depend on the regime which is determined by $\eta$. 
Observe that a multiplicative approximation for $\eta$ suffices since (i) in all three regimes, the complexity is linear in $\eta$ hence any fixed multiplicative factor applied to $\eta$ can be absorbed into the big-$\mathcal{O}$ constants, and (ii) adjacent regimes have the same complexity (up to constant factors) near the threshold, i.e., there is no ``discontinuity'' in the sample complexity with respect to $\eta$. Hence incorrect categorization of $\eta$-regime due to a multiplicative error still assigns a sample complexity that is equivalent to the sample complexity associated with the correct $\eta$-regime up to big-$\mathcal{O}$ constants. Finally, once we establish $\eta = 
\mathcal{O}(1/s^{*})$ is in the first regime, we need no further estimate of $\eta$ since $k$, $n$, $b$ are functions of $B$ and $\epsilon$, not $\eta$. 

Our information about $\eta$ comes from joint measurements, specifically from when they \emph{fail}. 
In Appendix~\ref{sec:combinatorics}, Theorem~\ref{thm:probability_of_success}, 
we show that the success probability of an $n$-sample measurement is bounded between $(1 - \eta)^{n-1}$ and $(1 - \eta)^{n-1}(1 + \mathcal O(\eta^2))$. However, for a $2$-sample measurement, we can be more specific: 
\begin{equation}
\eta \leq \eta + \eta^{2}/2 \leq 1-\Pr[\success] \leq \eta + \eta^2 \leq 2\eta. 
\end{equation}
In other words, $2$-sample measurements fail with some probability $p = \Theta(\eta)$ which is a multiplicative approximation for $\eta$. Hence, we can reduce to the problem of using many independent Bernoulli trials to estimate their failure probability. Indeed, we give an algorithm that does exactly this. 
\begin{restatable}{theorem}{estimateeta}
    \label{thm:estimateeta}
    Let $r \geq 1$ be an integer. There is an algorithm which estimates the failure probability $p$ of a Bernoulli trial, such that the algorithm (i) outputs a constant-factor multiplicative approximation of $p$, and (ii) makes $\bigo{r/p}$ samples of the Bernoulli trial, except with an $\exp(-\Theta(r))$ probability of failure.
\end{restatable} 
See Appendix~\ref{app:estimateeta} for the proof. We can get an arbitrarily low failure probability $\delta$ by taking $r = \bigo{\log(1/\delta})$, though Theorem~\ref{thm:sample_complexity_main} is stated for $\delta = \Omega(1)$ and hence $r = \bigo{1}$. 

Finally, in the first regime of Theorem~\ref{thm:sample_complexity_main}, $\eta = \bigo{1/s^{*}}$ can be small (or even zero!) and the algorithm in Theorem~\ref{thm:estimateeta} would use too many samples if run to completion. Instead, we cut it off at $\Omega(r/p) = \Omega(r/s^{*})$ samples, confident (w.p.\ $1 - \exp(-\Theta(r))$) that $\eta = \bigo{1/s^{*}}$, and then (conveniently) the parameters $k = b = 1$ and $n = \mathcal O(\sqrt{B}/\epsilon + 1/\epsilon^2)$ do not require an estimate of $\eta$. 

    \subsection{Comparison to alternative approaches}

    There are two natural strategies to compare against. First, the original classical shadows paper uses \citep{Huang2020} uses single-copy measurements which coincide with our single-copy measurement $\mathcal{M}_1$ (c.f.~\cref{defn:abbreviated_measurement}). Below we give the optimal sample complexity within our $knb$ framework when constrained to single-copy ($n=1$) measurements. 

    \begin{theorem}
    \label{thm:single_copy}
        With single-copy measurements and purification, we get sample complexity 
        \[
        s = \begin{cases} 
        \mathcal O (\frac{B}{\epsilon^2}) & \text{if $\eta \leq \epsilon$,} \\
        \mathcal O (\frac{B\eta}{\epsilon^3}) & \text{if $\eta \geq \epsilon$.}
        \end{cases}
        \]
    \end{theorem}

    Second, we can turn to purification \emph{before} repetition and averaging, i.e., set $b = 1$ in our framework. Again, the result is somewhat worse. 
    \begin{theorem}
    \label{thm:no_average}
        There is an algorithm which purifies and makes joint measurements (no averaging), having sample complexity 
        \begin{align*}
        s &= \begin{cases}
            s^{*} & \text{if $\eta \leq 1/s^{*}$,} \\
            \eta (s^{*})^{2} & \text{if $\eta \geq 1/s^{*}$,}
        \end{cases}
        &
        \text{where } s^{*} &= \frac{\sqrt{B}}{\epsilon} + \frac{1}{\epsilon^2}.
        \end{align*}
    \end{theorem}
    Since these fall within our $knb$ framework, they cannot be any better than Theorem~\ref{thm:sample_complexity_main}.
    % Comparison 1
    Theorem~\ref{thm:single_copy} matches the performance of \citep{Huang2020} initially, thus performing quadratically worse than Theorem~\ref{thm:sample_complexity_main} in the $B=\omega(1)$ regime. Theorem~\ref{thm:single_copy}'s performance then degrades by a factor of $\eta/\epsilon$ and compares poorly to the $\frac{B \eta + 1}{\epsilon^2}$ performance in the middle regime of Theorem~\ref{thm:sample_complexity_main}.
    % Comparison 2
    Theorem~\ref{thm:no_average} matches our performance for very small $\eta$, as expected, but then picks up a $\bigo{1/\epsilon^4}$ term which compares unfavorably with either the $\bigo{1/\epsilon^2}$ or $\bigo{1/\epsilon^{5/2}}$ terms in Theorem~\ref{thm:sample_complexity_main}.

% -------------------------------------------------------------
% ACKNOWLEDGEMENTS 
% -------------------------------------------------------------

\section*{Acknowledgements}
We thank Richard Kueng for useful discussions. 

\bibliography{bibliography}

\newpage
\appendix

    \section{Preliminaries}
    \label{sec:preliminaries}

    Let $\CC$ be the set of complex numbers and $\CCD$ the space of $d$-dimensional complex vectors. 
    \begin{defn}[Operator spaces]
    Given a Hilbert space $V$, let $\lin{V}$ denote the set of linear operators from $V$ to $V$. Let 
    \[
    \dens{V} = \{ A \in \lin{V} : \Tr(A) = 1, A \succeq 0 \}
    \]
    be the set of \emph{density matrices} which have trace $1$ and are positive semidefinite. 
    \end{defn}

    \subsection{Symmetric and exchangeable operators}

    Let $\symm_{n}$ denote the symmetric group of permutations on $n$ elements.
    
    \begin{defn}[permutation operator]
        Given a permutation $\pi \in \symm_n$ (for $n \geq 1$), define a permutation operator $W_{\pi} \in \CC^{d^{n} \times d^{n}}$ such that 
        $$
        W_{\pi} \ket{x_1} \cdots \ket{x_n} = \ket{x_{\pi^{-1}(1)}} \cdots \ket{x_{\pi^{-1}(n)}},
        $$
        and extend by linearity. 
        That is, $W_{\pi}$ acts on $(\CCD)^{\otimes n}$ by permuting the qudits, sending the qudit in position $i$ to position $\pi(i)$. 
    \end{defn}

    \begin{defn}[symmetric and exchangeable]
        Let $\rho \in \link{n}$. We say $\rho$ is \emph{exchangeable} if 
        $\rho = W_{\pi} \rho W_{\pi}^{-1}$ for all $\pi \in \symm_n$, and $\rho$ is \emph{symmetric} if $\rho = W_{\pi} \rho W_{\sigma}$ for all $\pi, \sigma \in \symm_n$.
    \end{defn}
    Put another way, exchangeable operators \emph{commute} with $W_{\pi}$ or $\sym^{(n)}$, whereas symmetric operators \emph{absorb} $W_{\pi}$ and $\sym^{(n)}$. 
    Naturally, an \emph{exchangeable state} or \emph{symmetric state} is a state (i.e., a density matrix) which is also exchangeable or symmetric respectively. 
    \begin{defn}[symmetric subspace]
        The \emph{symmetric subspace} of a system of $n$ qudits of dimension $d$ is the set of all symmetric operators.\footnote{One can check that symmetric operators are closed under addition and scalar multiplication, and thus this is subspace.} 
        Let $\symdim{d}{n}$ to denote the dimension of the symmetric subspace and define $\sym^{(n)}$ to be the projector onto it (notationally omitting the dependence on $d$, the dimension of the qudit, which is typically fixed for our purposes). 
    \end{defn}

    It's worth noting that states in the symmetric subspace are also exchangeable, but exchangeable are not necessarily in the symmetric subspace---take, for example, the maximally mixed state. We have two characterizations of the symmetric subspace.
    \begin{fact}
    \label{fact:symmetric_subspace_size}
        For all $n \geq 0$,
        $\sym^{(n)} = \frac{1}{n!} \sum_{\pi \in \symm_n} W_{\pi}$, and 
        $\symdim{d}{n} = \binom{n+d-1}{d-1}$.
    \end{fact}
    
    The integral of $\ketbra{\psi}{\psi}$ over the Haar measure is known from, e.g., \cite{scott2006tight}.

    \begin{lemma}
    \label{lem:symmetric_haar}
    $$
    \symdim{d}{n} \int_{\psi} \left( \ketbra{\psi}{\psi} \right)^{\otimes n} \mathrm{d}\psi = \sym^{(n)} = \frac{1}{n!} \sum_{\pi \in \symm_{n}} W_{\pi}
    $$  
    where $\sym^{(n)}$ is the projector onto the symmetric subspace and $W_{\pi}$ is the operator that permutes $n$ qudits by an $n$-element permutation $\pi \in \symm_n$.  
    \end{lemma}

    Recall that for a system $V = V_{A} \otimes V_{B}$, the \emph{partial trace} $\Tr_{A} \colon \lin{V_{A} \otimes V_{B}} \to \lin{V_{B}}$ is the unique superoperator such that 
    \[
        \Tr_{A}(\rho_{A} \otimes \rho_{B}) = \Tr(\rho_{A}) \rho_{B}
    \]
    for all $\rho_{A} \in \lin{V_{A}}$ and $\rho_{B} \in \lin{V_{B}}$. We use the notation $\Tr_{n \to k}$ to represent the map from $\link{n}$ to $\link{k}$ given by
    \begin{align}
        \Tr_{n \to k}(A) := \Tr_{1,\ldots,n-k}(\sym^{(n)} A)
        %\binom{n}{k}^{-1} \sum_{\substack{X \subseteq [n] \\ |X| = n-k}} \Tr_{X}(\sym^{(n)} \rho).
    \end{align}
    That is, $\Tr_{n \to k}$ reduces an $n$-qudit state to $k$ qudits symmetrically.

    \section{The Standard Symmetric Joint Measurement on almost-pure states}

    Let us recall the measurement used in pure state classical shadows \cite{gps2022sampleoptimal} and many other learning tasks.
    \begin{defn}
    \label{defn:measurement}
    The \emph{standard symmetric joint measurement} is a measurement on $n$ qudits. It is defined by the POVM $\mathcal{M}_{n} = \{ F_{\psi} \}_{\psi} \cup \{ F_{\fail} \}$ with elements
    $$
    F_{\psi} := \symdim{d}{n} \cdot \ketbra{\psi}{\psi}^{\otimes n} \mathrm{d} \psi,
    $$
    for all $d$-dimensional pure states $\psi$, proportional to the Haar measure $\mathrm{d} \psi$, plus a ``failure'' outcome $F_{\fail} := \eye - \sym^{(n)}$ for non-symmetric states.
    \end{defn}
    
    The measurement \emph{fails} if we get outcome $F_{\fail}$, otherwise we say it \emph{succeeds}. In what follows, we will let $\Psi$ be a random variable over $\lin{\CCD}$ representing the outcome of measuring $\rho^{\otimes n}$ with $\mathcal{M}_n$ where 
    \[
    \Psi = \begin{cases}
        0 & \text{for outcome $\fail$,} \\
        \ketbra{\psi}{\psi} & \text{for outcome $\psi$.}
    \end{cases}
    \]

    When the measurement succeeds, we construct an estimator $\hat\phi = \frac{(d+n) \Psi - \eye}{n}$ from this random variable. We quantify the performance of $\hat\phi$ in the theorem below, but it is close enough to $\Phi_1$ to be useful in solving the Principal Eigenstate Classical Shadows problem.

    \robustnessjoint*

    The proof of this theorem is one of the main technical contributions of this paper, and we spend the rest of this section proving it. Before we begin, let us give a detailed outline of the structure of the proof, to help orient the reader.

    First, the strategy of \emph{measuring} copies of a state and \emph{preparing} a fixed state for each outcome is known as a \emph{measure-and-prepare channel}. Our particular case---measure with $\mathcal{M}_n$ and prepare $\Psi$---is especially well studied. In Appendix~\ref{app:chiribella}, we adapt a result of Chiribella \cite{chiribella2011quantum}  (Theorem~\ref{thm:chiribella}) about this channel to get formulas for $\E[\Psi]$, $\E[\Psi^{\otimes 2}]$, and $\Var[\Psi]$ (Theorem~\ref{thm:moments_of_psi_given_ms}). 

    However, the formulas for $\E[\Psi]$, $\E[\Psi^{\otimes 2}]$, and $\Var[\Psi]$ are in terms two linear operators, 
    \begin{align*}
        M_1 &:= \frac{\Tr_{1,\ldots,n-1}(\sym^{(n)} \rho^{\otimes n})}{\Tr(\sym^{(n)} \rho^{\otimes n})}, &
        M_2 &:= \frac{\Tr_{1,\ldots,n-2}(\sym^{(n)} \rho^{\otimes n})}{\Tr(\sym^{(n)} \rho^{\otimes n})}.
    \end{align*} 
    While it is easy to write, e.g., $\Tr(\sym^{(n)}\rho^{\otimes n})$, it is not so easy to bound it in terms of the principal eigenvalue ($\lambda_1$) or the deviation ($\eta$). For instance, a natural approach is to expand $\sym^{(n)}$ as a sum of permutations (by Fact~\ref{fact:symmetric_subspace_size}), and calculate $\Tr(W_{\pi} \rho^{\otimes n}) = \prod_{C \in \pi} \Tr(\rho^{|C|})$ where the product is over cycles $C$ of $\pi$. This argument establishes that $\Tr(\sym^{(n)} \rho^{\otimes n})$ is a symmetric polynomial in the spectrum $\lambda_1, \ldots, \lambda_n$ of $\rho$, but the dependence on $\eta$ is difficult to bound. We need a different approach.

    Our solution is to expand $\rho^{\otimes n}$ in the basis $\{ \Phi_{a_1} \otimes \cdots \otimes \Phi_{a_n} \colon a_1, \ldots, a_n \in [d] \}$ where $\Phi_1, \ldots, \Phi_{n} \in \dens{\CCD}$ are the eigenvectors of $\rho$. Actually, since $\rho^{\otimes n}$ is symmetric, we are more concerned with the number of occurrences $e_i = \# \{ j : a_j = i \}$ of each factor $\Phi_i$, which constitute a \emph{type} $\ve = (e_1, \ldots, e_d)$ of basis state. We define $\sigma(\ve)$ as the average of all states of type $\ve$. Theorem~\ref{thm:ms_from_mse} shows that $M_1$ and $M_2$ are expectations of $M_1^{\ve}$ and $M_2^{\ve}$ below, where $\ve$ is sampled from some distribution $\truedist$. 
    \begin{align*}
        M_1^{\ve} &:= \frac{\Tr_{1,\ldots,n-1}(\sym^{(n)} \sigma(\ve))}{\Tr(\sym^{(n)} \sigma(\ve))}, &
        M_2^{\ve} &:= \frac{\Tr_{1,\ldots,n-2}(\sym^{(n)} \sigma(\ve))}{\Tr(\sym^{(n)} \sigma(\ve))}.
    \end{align*}

    We have reduced the computation of $M_1$ and $M_2$ to the computation of $M_1^{\ve}$ and $M_2^{\ve}$ for a vast set of $\ve$, but is this really progress? Yes (!), because the following trace is either $0$ or $1$, depending on the permutation $\pi$, and similar facts are true of the partial traces we need. 
    \begin{fact}
        \label{fact:count_trace_permutation}
        \[
        \Tr(W_{\pi} (\Phi_{a_1} \otimes \cdots \otimes \Phi_{a_n})) = 
        \begin{cases}
            0 & \text{if $a_{\pi(j)} \neq a_{j}$ for some $j \in [n]$,} \\
            1 & \text{otherwise.}
        \end{cases}
        \]
    \end{fact}
    This transforms the computation of $M_1^{\ve}$ and $M_2^{\ve}$ into a combinatorial problem: we can compute, e.g., $\Tr(\sym^{(n)} \sigma(\ve))$ by \emph{counting} how many permutations contribute $1$ (rather than $0$). Using this approach, we compute a probability of success, $Z^{\ve}$, in Theorem~\ref{thm:sigma0}, and $M_1^{\ve}$ and $M_2^{\ve}$ in Theorems~\ref{thm:sigma1} and \ref{thm:sigma2} respectively. 
    
    In principle, we now have concrete expressions for $M_1^{\ve}$ and $M_2^{\ve}$ in terms of $\ve$ and $\{ \Phi_i \}_{i=1}^{d}$, and it only remains to compute expectations over $\ve \samp \truedist$. Here the distribution $\truedist$ puts weight on $\ve$ proportional to $\lambda_1^{e_1} \cdots \lambda^{e_d}$. Despite the simplicity of this probability mass function, we could find no closed form for $\E_{\ve \in \truedist}[e_i]$ or $\E_{\ve \in \truedist}[e_i e_j]$. This last technical hurdle forces us to approximate: we introduce a distribution $\fakedist$ for which we \emph{can} compute $\E_{\ve \in \fakedist}[e_i]$ and $\E_{\ve \in \fakedist}[e_i e_j]$, and show that is it close to $\truedist$ when $\eta$ is sufficiently small. 

    We finish the proof by combining the pieces. The expectation of $\rhohat$, for example, is finally completed in Corollary~\ref{cor:estimator_mean_truedist}. The corollary uses various theorems to bound the distance to the expectation under $\fakedist$, which is given by Theorem~\ref{thm:estimator_mean_fakedist}. Theorem~\ref{thm:estimator_mean_fakedist} is derived from Corollary~\ref{cor:estimator_mean_and_variance_wrt_ms} to claim the expectation of the estimator (called $\mhat$ rather than $\rhohat$ in that section) is $M_1$, Theorem~\ref{thm:ms_from_mse} to expand $M_1$ as a distribution over $M_1^{\ve}$, Theorem~\ref{thm:sigma1} to turn $M_1^{\ve}$ into actual $e_i$ and $\Phi_i$ terms, and Proposition~\ref{prop:geometric_moments} to evaluate those for the geometric random variables defining $\fakedist$.
        
    \subsection{Chiribella's Theorem: Moments from Partial Traces} 
    \label{sec:chiribella}

    The measurement we use ($\mathcal{M}_n$) has applications to other pure state learning tasks, so there is already a result characterizing the moments of the outcome (i.e., $\E[\Psi^{\otimes k}]$), and in particular the mean and variance. Specifically, in Appendix~\ref{app:chiribella}, we take a result of Chiribella \cite{chiribella2011quantum} and repackage it into the following. 
    \begin{restatable}{theorem}{chiribella}
        \label{thm:alternativeMP}
        Fix integers $n, k \geq 0$, let $A \in \densk{n}$ be an exchangeable $n$-qudit state, and let $\Psi$ be the outcome of measuring $A$ with $\mathcal{M}_n$.
        \begin{equation}
        \E[\Psi^{\otimes k} \mid \success] = \frac{1}{\symdim{(d+n)}{k}} \sym^{(k)} \parens*{\sum_{s=0}^{k} \binom{n}{s} \binom{k}{s} \parens*{\frac{\Tr_{n \to s}(A)}{\Tr_{n \to 0}(A)} \otimes \eye^{\otimes k-s}} }\sym^{(k)}
        \label{eq:chiribella}
        \end{equation}
    \end{restatable} 
    
    In other words, to compute $\E[\Psi^{\otimes k} \mid \success]$ we only need $\Tr_{n \to s}(A) = \Tr_{[n-s]}(\sym^{(n)}A)$ for $0 \leq s \leq k$. More specifically, the important data about the state are $M_0(A), \ldots, M_k(A)$ where we define the function\footnote{
    This is technically a partial function because of potential division by $0$, but $\Tr_{n \to 0}(A) = \Tr(\sym^{(n)} A) = \Pr[\success]$, so this is only a problem if the measurement always fails, and then we have bigger problems.}
    \[
    M_k(A) := \frac{\Tr_{n \to k}(A)}{\Tr_{n \to 0}(A)}.
    \]
    We further abbreviate $M_k(A)$ to just $M_k$ when the state is understood. 
    
    Since we aim to compute the mean and variance of $\Psi$, i.e., $\E[\Psi \mid \success]$ and 
    \[
    \Var[\Psi \mid \success] = \E[\Psi \mid \success^{\otimes 2}] - \E[\Psi \mid \success]^{\otimes 2}, 
    \]
    we only apply Theorem~\ref{thm:alternativeMP} with $k = 1, 2$. Below, we specialize Theorem~\ref{thm:alternativeMP} to these two cases, using our new $M_1$, $M_2$ notation. 
    \begin{theorem}
        \label{thm:moments_of_psi_given_ms}
        Let $A$ be an exchangeable $n$-qudit state, and let $\Psi$ be a random variable for the outcome of measuring $A$ with $\mathcal{M}_n$. 
        \begin{align}
            \E[ \Psi \mid \success] &= \frac{\eye + n M_1}{d + n}, \label{eq:mean_of_psi} \\
            \E[ \Psi^{\otimes 2} \mid \success] &= \frac{2 \sym^{(2)}}{(d+n)(d+n+1)} \parens*{(\eye + nM_1)^{\otimes 2} + \binom{n}{2} M_2 - n^2 M_1^{\otimes 2}},
        \end{align}
        \begin{align}
            \Var[ \Psi \mid \success] &= \frac{W_{(1\,2)}\parens*{\eye + nM_1 \otimes \eye + n \eye \otimes M_1}}{(d+n)(d+n+1)} + \frac{n(n-1) M_2 - n^2 M_1^{\otimes 2}}{(d+n)(d+n+1)} - \frac{(\eye + nM_1)^{\otimes 2}}{(d+n)^2(d+n+1)}, \label{eq:variance_of_psi}
        \end{align}
        where the expectation and variance are over the randomness in the measurement.
    \end{theorem}
    \begin{proof}
        $\E[\Psi \mid \success]$ and $\E[\Psi^{\otimes 2} \mid \success]$ are immediate from Theorem~\ref{thm:alternativeMP}. For the variance, we start with the definition:
        \begin{align*}
        \Var[\Psi \mid \success] 
        &= \E[\Psi^{\otimes 2} \mid \success] - \E[\Psi \mid \success]^{\otimes 2} \\
        &= \frac{2 \sym^{(2)}}{(d+n)(d+n+1)} \parens*{(\eye + nM_1)^{\otimes 2} + \binom{n}{2} M_2 - n^2 M_1^{\otimes 2}} - \parens*{\frac{\eye + n M_1}{d + n}}^{\otimes 2}.
        \end{align*}
        Recall that $2 \sym^{(n)} = W_{(1)(2)} + W_{(1\,2)}$, so first and foremost there is a near-cancellation of two terms:
        \[
            \frac{1}{(d+n)(d+n+1)} W_{(1)(2)}(\eye + nM_1)^{\otimes 2} -  \parens*{\frac{\eye + n M_1}{d + n}}^{\otimes 2} = -\frac{(\eye + nM_1)^{\otimes 2}}{(d+n)^2(d+n+1)}.
        \]
        The remaining terms are 
        \[
        \frac{1}{(d+n)(d+n+1)} \parens*{\binom{n}{2} M_2 - n^2 M_1^{\otimes 2}} + \frac{W_{(1\,2)}}{(d+n)(d+n+1)} \parens*{(\eye + nM_1)^{\otimes 2} + \binom{n}{2} M_2 - n^2 M_1^{\otimes 2}}.
        \]
        Among the $W_{(1\,2)}$ terms, $(\eye + nM_1)^{\otimes 2}$ cancels with $n^2 M_1^{\otimes 2}$ leaving 
        \[
        \frac{W_{(1\,2)}}{(d+n)(d+n+1)} \parens*{\eye + nM_1 \otimes \eye + n \eye \otimes M_1}.
        \]
        Moreover, $M_2$ is invariant under $W_{(1\,2)}$ (i.e., $M_2 = W_{(1\,2)} M_2$), the two $M_2$ terms combine into $\frac{n(n-1)}{(d+n)(d+n+1)} M_2$, which we group with the left over $-n^2 M_1^{\otimes 2}$ term. 

        Altogether, the variance is 
        \[
            \Var[\Psi \mid \success] = \frac{W_{(1\,2)}\parens*{\eye + nM_1 \otimes \eye + n \eye \otimes M_1}}{(d+n)(d+n+1)} + \frac{n(n-1) M_2 - n^2 M_1^{\otimes 2}}{(d+n)(d+n+1)} - \frac{(\eye + nM_1)^{\otimes 2}}{(d+n)^2(d+n+1)}.
        \]
    \end{proof}
    
    \subsection{\texorpdfstring{Estimator $\mhat$}{Estimator from joint symmetric measurement}}
    \label{sec:estimator}

    Consider the mean of $\Psi$,
    \[
        \E[ \Psi \mid \success] = \frac{\eye + n M_1}{d + n},
    \]
    from Theorem~\ref{thm:moments_of_psi_given_ms}. This is a convex combination of $\eye / d$, the maximally mixed state, and $M_1$, which we observe is also a state.
    \begin{lemma}
        \label{lem:m_is_a_state}
        For any exchangeable state $A \in \densk{n}$ such that $\Tr_{n \to 0}(A) \neq 0$, we have $M_k(A) \in \densk{k}$ for all $k \geq 0$.
    \end{lemma}
    \begin{proof}
         If $M_k(A)$ is not positive semi-definite, then there is a state $\sigma$ that witnesses its negativity (i.e., $\Tr(M_k(A) \sigma) < 0$). Then $\eye \otimes \sigma$ witnesses the negativity of $\sym^{(n)}A$ ($\Tr((\sym^{(n)} A)(\eye \otimes \sigma)) < 0$), a contradiction. 

        For the trace, observe that 
         \[
            \Tr(M_k(A)) = \frac{\Tr(\Tr_{[n-k]}(\sym^{(n)}A))}{\Tr_{n \to 0}(A)} = \frac{\Tr(\sym^{(n)}A)}{\Tr_{n \to 0}(A)}= \frac{\Tr_{n \to 0}(A)}{\Tr_{n \to 0}(A)} = 1.
         \]
    \end{proof}
    In any case, clearly $M_1$ is the part of the estimator that we are interested in, the ``signal'' among the ``noise''. It is standard practice in classical shadows protocols to ``invert'' the channel to isolate the component of interest (in our case, $M_1$). That is, define an estimator\footnote{Often this estimator would be called $\rhohat$, and later we will rename it $\hat\phi$, but within this section we will use $\mhat$ since it is an unbiased estimator for $M_1$.} $\mhat := \tfrac{1}{n} \bracks*{(d+n) \Psi - \eye}$ so that 
    \[
    \E[\mhat] = \tfrac{1}{n} \bracks*{(d+n) \E[\Psi \mid \success] - \eye} = M_1. 
    \]
    It is important that \textit{we only define $\mhat$ when the measurement succeeds}---we throw away any failed measurements.\footnote{It will be important later how \emph{often} failure occurs, and this is explored in Theorem~\ref{thm:probability_of_success}.}  
    
    \begin{cor}[Estimator mean and variance in terms of $M_1$, $M_2$]
    \label{cor:estimator_mean_and_variance_wrt_ms}
        Suppose we measure a state $A \in \densk{n}$, the measurement succeeds, and we produce estimator $\mhat$ as described. Let $O$ be a Hermitian observable. The expectation and variance over the randomness of the measurement outcome,
        \begin{align}
            \E_{\meas}[\mhat] &= M_1 \\
            \Var_{\meas}[\Tr(O\mhat)] &\leq \frac{\Tr(O^2)}{n^2} + \frac{2 \inftynorm{O}^2}{n} + \frac{n-1}{n} \Tr(O^{\otimes 2} M_2) - \Tr(O M_1)^2, \label{eq:variance_of_traceomhat}
        \end{align}
        are functions of $M_1 := M_1(A)$, $M_2 := M_2(A)$. 
    \end{cor}
    \begin{proof}
        We have already seen $\E[\mhat]$, and 
        \[
        \Var[\mhat] = \Var \bracks*{\tfrac{1}{n} \bracks*{(d+n) \Psi - \eye}} = \Var[\tfrac{d+n}{n} \Psi] = \tfrac{(d+n)^2}{n^2} \Var[\Psi].
        \]
        We can substitute in \eqref{eq:variance_of_psi} from Theorem~\ref{thm:moments_of_psi_given_ms} to get an expression for $\Var_{\meas}[\mhat]$. 
        \[
            \Var_{\meas}[\mhat] = \frac{d+n}{d+n+1} \parens*{ \frac{W_{(1\,2)}\parens*{\eye + nM_1 \otimes \eye + n \eye \otimes M_1}}{n^2} + \frac{n-1}{n} M_2 - M_1^{\otimes 2} - \frac{(\eye + nM_1)^{\otimes 2}}{n^2(d+n)^2}}
        \]
        We remind the reader that $\Var[\mhat]$ is a $2$-qudit linear operator, and to get the variance of $\Tr(O \mhat)$, we need to take the trace with $O^{\otimes 2}$ since
        \begin{align*}
            \Var[\Tr(O \mhat)] &= \E[\Tr(O \mhat)^2] - \E[\Tr(O \mhat)]^2 \\
            &= \Tr(O^{\otimes 2} \E[\mhat^{\otimes 2}]) - \Tr(O^{\otimes 2} \E[\mhat]^{\otimes 2}) \\
            &= \Tr \parens*{O^{\otimes 2} \Var[\mhat]}.
        \end{align*}
        We therefore take the trace with $O^{\otimes 2}$ and make a few simplifications -- dropping negative terms, rounding $\frac{d+n}{d+n+1}$ up to $1$, and so on. 
        \begin{align*}
            & \qquad \Var[\Tr(O \mhat)] \\
            &= \Tr \parens*{O^{\otimes 2}\tfrac{d+n}{d+n+1} \parens*{ \frac{W_{(1\,2)}\parens*{\eye + nM_1 \otimes \eye + n \eye \otimes M_1}}{n^2} + \frac{n-1}{n} M_2 - M_1^{\otimes 2} - \frac{(\eye + nM_1)^{\otimes 2}}{n^2(d+n)^2}}} \\
            &= \tfrac{d+n}{d+n+1} \parens*{\frac{\Tr(O^2)}{n^2} + \frac{2\Tr(O^2 M_1)}{n} + \frac{n-1}{n} \Tr(O^{\otimes 2} M_2) - \Tr(O M_1)^2 - \parens*{\frac{\Tr(O) + n \Tr(OM_1)}{n(d+n)}}^2 } \\
            &\leq \frac{\Tr(O^2)}{n^2} + \frac{2 \Tr(O^2 M_1)}{n} + \frac{n-1}{n} \Tr(O^{\otimes 2} M_2) - \Tr(O M_1)^2. 
        \end{align*}
        Last, we apply H\"{o}lder's inequality for Schatten norms to bound $\Tr(O^2 M_1) \leq \inftynorm{O}^2 \norm{M_1}_1 = \inftynorm{O}^{2}$, using the fact that $\norm{M_1}_1 = 1$ because $M_1$ is a state (Lemma~\ref{lem:m_is_a_state}).
    \end{proof}

    \subsubsection{Pure State Classical Shadows}

    As a quick exercise, we can re-derive the mean and variance of the pure state classical shadows estimator from \cite{gps2022sampleoptimal}. When $\rho$ is pure, it is not hard to see that $\Tr_{n \to k}(\rho^{\otimes n}) = \rho^{\otimes k}$ and thus 
    \begin{align*}
        M_1 &= \frac{\Tr_{n \to 1}(\rho^{\otimes n})}{\Tr_{n \to 0}(\rho^{\otimes n})} = \frac{\rho}{1} = \rho, &  M_2 &= \frac{\Tr_{n \to 2}(\rho^{\otimes n})}{\Tr_{n \to 0}(\rho^{\otimes n})} = \frac{\rho^{\otimes 2}}{1} = \rho^{\otimes 2}.
    \end{align*}
    Thus $\mhat$ is an unbiased estimator for $M_1 = \rho$ (which is generally \emph{not} the case when $\rho$ is mixed). 
    \begin{lemma}
        Let $O$ be a Hermitian observable. When $\rho$ is pure, $\E[\Tr(O\mhat)] = \Tr(O\rho)$ and 
        \begin{align*}
        \Var_{\meas}[\Tr(O\mhat)] &\leq \frac{\Tr(O^2)}{n^2} + \frac{2 \inftynorm{O}^2}{n}.
        \end{align*}
    \end{lemma}
    \begin{proof}
        Use Corollary~\ref{cor:estimator_mean_and_variance_wrt_ms}. The expectation of $\Tr(O \mhat)$ follows immediately. For the variance, we observe that $\Tr(O^{\otimes 2}M_2) = \Tr(O^{\otimes 2} \rho^{\otimes 2}) = \Tr(O \rho)^{2}$, which is then more than cancelled out by $\Tr(O M_1)^{2} = \Tr( O \rho )^{2}$.
    \end{proof}

    \subsection{Classical mixture of orthogonal tensor products} 
    \label{sec:mixture}

    Corollary~\ref{cor:estimator_mean_and_variance_wrt_ms} reduces the mean and variance calculation to computing $M_1$ and $M_2$. However, even for $n = 4$, $\Tr_{4 \to 1}(\rho^{\otimes 4})$ is the unwieldy polynomial
    \[
        \frac{1}{24} \parens*{\rho (1 + 3\Tr(\rho^2) + 2\Tr(\rho^3)) + \rho^2(3 + \Tr(\rho^2)) + 6\rho^{3} + 6\rho^{4}}, 
    \]
    and $\Tr_{4 \to 2}(\rho^{\otimes 4})$ is even worse. It is hard to say much about, e.g., the overlap of $M_1$ with the leading eigenvector ($\Phi_1$), beyond the fact that it is some symmetric polynomial in the eigenvalues of $\rho$. It is unclear how to bound it in terms of the deviation, $\eta = 1 - \lambda_1$, especially for arbitrary $n$. We need a different approach to compute $M_1$ and $M_2$ for $\rho^{\otimes n}$. 

    Recall that $\rho = \sum_{i=1}^{d} \lambda_i \Phi_i$ where $\Phi_i := \ketbra{\phi_i}{\phi_i}$ are projectors onto the eigenvectors, $\ket{\phi_i}$, which form an orthonormal basis. It follows that we can expand $\rho^{\otimes n}$ in a basis of $n$-fold tensor products, $\Phi_{\va} := \otimes_{i=1}^{n} \Phi_{a_i}$, where $\va = (a_1, \ldots, a_n) \in [d]^{n}$. That is,
    \[
        \rho^{\otimes n} = \sum_{\va \in [d]} \parens*{\prod_{i=1}^{n} \lambda_{a_i}} \Phi_{\va}
    \]
    Since $\rho^{\otimes n}$ is exchangeable, we can express the right hand side as a sum of exchangeable operators. Specifically, we group the terms by type, where the \emph{type} of $\Phi_{\va}$ is a vector $\ve = (e_1, \ldots, e_d)$ where $e_i := \# \{ j : a_j = i \}$ is the number of occurrences of $\Phi_i$. The grouped states we call $\sigma(\ve)$.
	\begin{defn}
		Fix a basis $\ket{\phi_1}, \dotsc, \ket{\phi_d}$, and let $\Phi_i = \ketbra{\phi_i}{\phi_i}$ for all $i = 1, \ldots, d$. Given $\ve = (e_1, \dotsc, e_d) \in \mathbb N^{d}$ such that $e_1 + \dotsm + e_d = n$, define a mixed state 
		\[
		\sigma(e_1,\dotsc,e_d) := \frac{1}{n!} \sum_{\pi \in \symm_n} W_{\pi} \left( \bigotimes_{j=1}^{d} \Phi_j^{\otimes e_j}\right) W_{\pi}^{\dag} \in \densk{n}.
		\]
	\end{defn}
    Let us define some shorthand notation with $\ve$ for later use. 
	\begin{align*}
		\ve! &:= e_1! \dotsm e_d!, & \binom{n}{\ve} &:= \frac{n!}{\ve!}, & \vlambda^{\ve} &:= \lambda_1^{e_1} \dotsm \lambda_d^{e_d}, & %\sume = e_1 + \dotsm + e_d.
	\end{align*}
    We can now succinctly write $\rho^{\otimes n}$ as a convex combination of $\sigma(\ve)$. 
    \begin{prop} 
    \label{prop:rho_as_mixture}
    \[
    \rho^{\otimes n} = \sum_{\ve} \binom{n}{\ve} \vlambda^{\ve} \sigma(\ve)
    \]
    \end{prop}

    A mixture of this form is indistinguishable from a $\sigma(\ve)$ selected at random from the distribution with mass $\Pr[\ve] = \binom{n}{\ve} \vlambda^{\ve}$, which we recognize as a multinomial distribution with $n$ trials for $d$ events with probabilities $\lambda_1, \ldots, \lambda_d$. However, the pertinent distribution for our calculations, which we name $\truedist$, is this multinomial conditioned on successful measurement, since $\mhat$ is conditional on a successful measurement outcome. That is, the probability mass function for for $\truedist$ is 
    \[
    \Pr[\ve \mid \success] = \frac{\Pr[\success \mid \ve] \Pr[\ve]}{\Pr[\success]} \propto \Pr[\success \mid \ve] \cdot \binom{n}{\ve} \vlambda^{\ve}.
    \]
    We will calculate shortly (Theorem~\ref{thm:sigma0}), the probability of success for a given $\ve$, to make this distribution more explicit.

    Recall that the purpose of decomposing $\rho^{\otimes n}$ as a mixture (Proposition~\ref{prop:rho_as_mixture}) was to give another path to compute $M_1$ and $M_2$. To this end, we define 
    \begin{align*}
        M_k^{\ve} &:= \frac{\Tr_{n \to k}(\sigma(\ve))}{\Tr_{n \to 0}(\sigma(\ve))},
    \end{align*}
    for all $k \geq 1$, and prove the following connection to the original $M_k$s.
    \begin{thm}
        \label{thm:ms_from_mse}
        Fix $0 \leq k \leq n$, and then $M_k(\rho^{\otimes n}) = \E_{\ve \samp \truedist}[ M_{k}(\sigma(\ve))]$, i.e., $M_k(\rho^{\otimes n})$ is the expectation of $M_k(\sigma(\ve))$ over $\ve$ sampled from $\truedist$. 
    \end{thm}
    \begin{proof}
        \begin{align*}
        M_k(\rho^{\otimes n}) 
        &= \frac{\Tr_{n \to k}(\rho^{\otimes n})}{\Tr_{n \to 0}(\rho^{\otimes n})} & \text{definition} \\
        &= \frac{\Tr_{n \to k}(\rho^{\otimes n})}{\Pr[\success]} & \text{because $\Pr[\success] = \Tr_{n \to 0}(\rho^{\otimes n})$} \\
        &= \frac{1}{\Pr[\success]} \sum_{\ve} \Pr[\ve] \cdot \Tr_{n \to k}(\sigma(\ve)) & \text{linearity of trace} \\
        &= \sum_{\ve} \frac{\Pr[\ve]}{\Pr[\success]} \cdot \Tr_{n \to k}(\sigma(\ve)) \cdot \frac{\Pr[\success \mid \ve]}{\Tr_{n \to 0}(\sigma(\ve))} & \text{since $\Pr[\success \mid \ve] = \Tr_{n \to 0}(\sigma(\ve))$} \\ 
        &= \sum_{\ve} \frac{\Pr[\success \mid \ve] \Pr[\ve]}{\Pr[\success]} \cdot \frac{\Tr_{n \to k}(\sigma(\ve))}{\Tr_{n \to 0}(\sigma(\ve))} & \text{rearrange} \\ 
        &= \sum_{\ve} \Pr[\ve \mid \success] \cdot M_k(\sigma(\ve)) & \text{Bayes' rule} \\
        &= \E[ M_{k}(\sigma(\ve)) \mid \success].
        \end{align*}
    \end{proof}
    That is, $M_k = \E[M_k^{\ve}]$. We stress that there are now two sources of randomness affecting our estimator: \emph{mixture randomness} arising from a random choice of initial state $\sigma(\ve)$, and \emph{measurement randomness} caused by the inherent randomness of measuring a quantum state.
    
    \subsection{\texorpdfstring{$M_1$}{M1} and \texorpdfstring{$M_2$}{M2} for \texorpdfstring{$\sigma(\ve)$}{sigma(e)}}
    \label{sec:combinatorics}

    Section~\ref{sec:chiribella} and Section~\ref{sec:estimator} showed that moments of $\Psi$ and $\mhat$ can be calculated from $M_1$ and $M_2$, then Section~\ref{sec:mixture} expressed $M_1$ and $M_2$ as expectations of $M_1^{\ve}$ and $M_2^{\ve}$. In this section, we calculate 
    \[
    Z^{\ve} = \Pr[\success \mid \ve] = \Tr_{n \to 0}(\sigma(\ve)),
    \]
    as well as $M_1^{\ve}$, and $M_2^{\ve}$. 
    
    As a starting point for all three calculations, we have that
    \[
    \Tr_{n \to k}(\sigma(\ve)) = \Tr_{[n-k]}(\sym^{(n)}\sigma(\ve)) = \Tr_{[n-k]}(\sym^{(n)} \Phi^{\ve})= \frac{1}{n!} \sum_{\pi \in \symm_n} \Tr_{[n-k]}(W_{\pi} \Phi^{\ve})
    \]
    where we define $\Phi^{\otimes \ve} = \bigotimes_{i=1}^{d} \Phi_j^{\otimes e_j}$.

    We claim that expressions like $\Tr_{[n-k]}(W_{\pi} \Phi^{\otimes \ve})$ can be evaluated combinatorially --- the permutation will loop through various $\Phi_i$ tensor factors of $\Phi^{\otimes \ve}$, and if any two adjacent $\Phi_i$ and $\Phi_j$ are orthogonal (we call this \emph{mixing eigenstates} in the theorems that follow), then the entire term vanishes. If not, then many of the $\Phi_i$ are traced out, and those that remain appear in the result. 

    The following proofs may be a bit opaque if the reader cannot visualize $\Tr_{[n-k]}(W_{\pi} \Phi^{\otimes \ve})$ in the language of tensor networks. Tensor networks are a graphical model of linear operators where, e.g., the permutation operator $W_{\pi}$ is drawn as a literal permutation of legs, and partial trace is achieved by looping the output legs for the traced-out qudits back to the input legs. We cannot include a full introduction to tensor networks here, but refer the reader to either \cite{roberts2017chaos} or \cite{gps2022sampleoptimal} to see examples of these ideas in action on very similar problems.

    With that, we begin by calculating the full trace. Since $\sym^{(n)}$ is the sum of the successful measurement outcomes, it also represents the probability of success.
    \begin{thm}
        \label{thm:sigma0}
        For all $n \geq 1$ and $\ve$ such that $e_1 + \cdots + e_d = n$, 
        \[
        Z^{\ve} := \Tr_{n \to 0}(\sigma(\ve)) = \Tr( \sym^{(n)} \sigma(\ve)) = \binom{n}{\ve}^{-1}.
        \]
    \end{thm}
    \begin{proof}
        Expand $\sigma$ with the definition and $\sym^{(n)}$ as an average over permutations:
		\begin{align*}
			\Tr\left( \sym^{(n)} \sigma(\ve) \right)
            = \frac{1}{(n!)^2} \sum_{\pi,\sigma \in \symm_n} \Tr\left( W_{\pi} \left(\bigotimes_{j=1}^{d} \Phi_j^{\otimes e_j} \right) W_{\pi}^{\dag} W_\sigma \right)
            = \frac{1}{n!} \sum_{\pi \in \symm_{n}} \Tr\left( \Phi_1^{\otimes e_1} \dotsm \Phi_d^{\otimes e_d} W_{\pi} \right).
		\end{align*}
		Any permutation $\pi$ which \emph{mixes eigenstates} by having a cycle which involves $\Phi_i$ and $\Phi_j$ for $i \neq j$ may be ignored, since the trace factor for that cycle will be zero. If a permutation does not mix eigenstates, then we have 
		$
		\Tr\left( \Phi_1^{\otimes e_1} \dotsm \Phi_d^{\otimes e_d} W_{\pi} \right) = 1,
		$
		since it is a product of traces of the form $\Tr(\Phi_i^{k})$, each of which is $1$ because $\Phi_i$ is pure.
		
		In other words, it suffices to count permutations which do not mix eigenstates. Clearly there are $e_i!$ ways to permute the $\Phi_i^{\otimes e_i}$ factors among themselves, and we make this choice independently for each $i$ for a total of $\ve!$ permutations which give unit trace. Therefore,
        \[
            \Tr\left( \sym^{(n)} \sigma(\ve) \right) = \frac{\ve!}{n!}= \binom{n}{\ve}^{-1},
        \]
        completing the proof.
    \end{proof}

    Next, we find that $M_1^{\ve}$ has a surprisingly clean form.
    \begin{thm}
        \label{thm:sigma1}
        For all $n \geq 1$ and $\ve$ such that $e_1 + \cdots + e_d = n$, 
        \[
        M_1^{\ve} := \frac{\Tr_{n \to 1}(\sigma(\ve))}{\Tr_{n \to 0}(\sigma(\ve))} = \frac{\Tr_{1,\ldots,n-1}( \sym^{(n)} \sigma(\ve))}{Z^{\ve}} = \frac{1}{n} \sum_{i=1}^{d} e_i \Phi_i.
        \]
    \end{thm}
    \begin{proof}
        For any exchangeable $A \in \densk{n}$, we have 
        \[
        \Tr_{[n-1]}(A) = \Tr_{-1}(A) = \Tr_{-i}(W_\pi A W_{\pi^{-1}}) = \Tr_{-i}(A)
        \]
        for any $\pi \in \symm_n$ such that $\pi(i) = n$. Therefore, we can expand $\sym^{(n)}$ as an average over permutations, and also average over the indices which are traced out to obtain
        \[
            \Tr_{n \to 1}(\sigma(\ve)) = \Tr_{[n-1]}( \sym^{(n)} \sigma(\ve)) = \frac{1}{n} \sum_{i=1}^{n} \frac{1}{n!} \sum_{\pi \in \symm_n} \Tr_{-i}\parens*{ W_{\pi} \bigotimes_{j=1}^{d} \Phi_j^{\otimes e_j} }
        \]
        As before, if $\pi$ mixes eigenstates $\Phi_i$ and $\Phi_j$ in the same cycle (for $i \neq j$), then that term contributes $0$. There are $e_1! \dotsm e_d!$ permutations which do not mix eigenstates. Suppose index $i$ corresponds to $\Phi_j$. In each of the non-mixing permutations, the partial trace is $\Phi_j$. Since there are $e_j$ indices where there is a $\Phi_j$, it follows that 
        \[
            \Tr_{n \to 1}(\sigma(\ve)) = \binom{n}{\ve}^{-1} \frac{1}{n} \sum_{i=1}^{d} e_i \Phi_i.
        \]
        Dividing through by $Z^{\ve} = \Tr_{n \to 0}(\sigma(\ve)) = \binom{n}{\ve}^{-1}$ finishes the proof.
    \end{proof}

    \begin{thm}
        \label{thm:sigma2}
        For all $n \geq 1$ and $\ve$ such that $e_1 + \cdots + e_d = n$, 
        \[
        M_2^{\ve} = 
        \frac{\Tr_{n \to 1}(\sigma(\ve))}{\Tr_{n \to 0}(\sigma(\ve))} = \frac{\Tr_{1,\ldots,n-1}( \sym^{(n)} \sigma(\ve))}{Z^{\ve}} = \frac{2\sym^{(2)}}{n(n-1)}
        \parens*{\sum_{i \neq j} e_i e_j \Phi_i \otimes \Phi_j + \sum_{k} \binom{e_k}{2} \Phi_k^{\otimes 2}}.
        \]
    \end{thm}
    \begin{proof}
        Expand $\sym^{(n)}$ as an average over permutations $\pi$, and average over an ordered pair of distinct indices, $a$ and $b$.
        \[
            \Tr_{[n-2]}( \sym^{(n)} \sigma(\ve)) = \frac{1}{n(n-1)} \sum_{a \neq b} \frac{1}{n!} \sum_{\pi \in \symm_n} \Tr_{-a,-b}\parens*{ W_{\pi} \bigotimes_{k=1}^{d} \Phi_k^{\otimes e_k} }
        \]
        For each $a, b$, divide the permutations into those where $a, b$ are in separate cycles (type \Romannum{1}), and those where $a, b$ are in the same cycle (type \Romannum{2}). We note that composing a permutation with the transposition $(a \, b)$ also changes the type, and since this operation is clearly invertible, it is a bijection between the two types of permutations. That is, there are equally many type \Romannum{1} and type \Romannum{2} permutations. In fact, we can pair up the permutations $\pi, \pi'$ (one of each type) matched by the bijection, factor $(W_{(a)(b)} + W_{(a\,b)})$ out of $W_{\pi} + W_{\pi'}$ and out of the partial trace where it becomes $2 \sym^{(2)}$. In other words, it suffices to analyze type \Romannum{1} permutations and then multiply by $2 \sym^{(2)}$.

        Let us say a position $a$ is \emph{incident} to state $\Phi_i$ if $\Phi_i$ is the $a$th term of the tensor product. The first case we are interested in is when $a$ and $b$ are incident to the same state, $\Phi_i$. We observe that there are $e_1! \cdots e_d!$ non-mixing permutations total, but this includes both type \Romannum{1} and type \Romannum{2} permutations. There are $e_i (e_i - 1)$ ways to pick $a$ and $b$ from $\Phi_i^{\otimes e_i}$, but we divide this in half to get those of type \Romannum{1}. Hence, when $a$ and $b$ are incident to the same $\Phi_i$, we have
        \[
             \sum_{\pi \in \textrm{Type \Romannum{1}}} \Tr_{-a,-b}\parens*{ W_{\pi} \bigotimes_{k=1}^{d} \Phi_k^{\otimes e_k} } 
             = \ve! \binom{e_i}{2} \Phi_i \otimes \Phi_i.
        \]
        The second case is $a$ incident to $\Phi_i$ and $b$ incident to $\Phi_j$ for $i \neq j$. There are $e_1! \cdots e_d!$ type \Romannum{1} permutations which do not mix eigenstates, and thus have a nonzero partial trace to contribute to the sum. Note that $\ve!$ over $n!$ from the average over permutations gives the $\binom{n}{\ve}^{-1}$ we've come to expect in these calculations. There are $e_i$ indices $a$ within $\Phi_i^{\otimes e_i}$, and $e_j$ indices $b$ within $\Phi_j^{\otimes e_j}$. To summarize, we get
        \[
             \sum_{\pi \in \textrm{Type \Romannum{1}}} \Tr_{-a,-b}\parens*{ W_{\pi} \bigotimes_{k=1}^{d} \Phi_k^{\otimes e_k} } 
             = \ve! \cdot e_i e_j \Phi_i \otimes \Phi_j.
        \]
        Altogether, this leads to
        \begin{align*}
             \Tr_{[n-2]}( \sym^{(n)} \sigma(\ve)) &= \frac{1}{n(n-1)} \sum_{a \neq b} \frac{1}{n!} \sum_{\pi \in \symm_n} \Tr_{-a,-b}\parens*{ W_{\pi} \bigotimes_{k=1}^{d} \Phi_k^{\otimes e_k} } \\
             &= \frac{2\sym^{(2)}}{n(n-1)} \sum_{a \neq b} \frac{1}{n!} \sum_{\pi \in \textrm{Type \Romannum{1}}} \Tr_{-a,-b}\parens*{ W_{\pi} \bigotimes_{k=1}^{d} \Phi_k^{\otimes e_k} } \\
             &= \frac{2\sym^{(2)}}{n(n-1)} \cdot \frac{\ve!}{n!} \left(\sum_i \binom{e_i}{2} \Phi_i \otimes \Phi_i + \sum_{i \neq j} e_i e_j \Phi_i \otimes \Phi_j \right),
        \end{align*}
        and dividing by $Z^{\ve} = \Tr_{n \to 0}(\sigma(\ve)) = \binom{n}{\ve}^{-1}$ yields the stated result. 
    \end{proof}

    Let us combine and summarize Theorems~\ref{thm:sigma0},~\ref{thm:sigma1}, and~\ref{thm:sigma2} in one result. 
    \begin{cor}
    \label{cor:symmetric_subspace_components}
        For all $n \geq 1$,
        \begin{align*}
            Z := \Tr \parens*{\sym^{(n)} \rho^{\otimes n}} &= \sum_{\ve} \vlambda^{\ve}, \\
            M_1 &= \frac{1}{Z} \sum_{\ve} \vlambda^{\ve} \frac{1}{n} \parens*{\sum_{i=1}^{d} e_i \Phi_i}, \\
            M_2 &= \frac{1}{Z} \sum_{\ve} \vlambda^{\ve} 
        \frac{2\sym^{(2)}}{n(n-1)}
        \parens*{\sum_{i \neq j} e_i e_j \Phi_i \otimes \Phi_j + \frac{1}{2} \sum_{k} e_k(e_k-1) \Phi_k \otimes \Phi_k},
        \end{align*}
        where all three outer sums are over $\ve \in \mathbb N^{d}$ such that $e_1 + \cdots + e_d = n$.
    \end{cor}
    \begin{proof}
        For all $k$ we have,
        \[
        \Tr_{n \to k}(\rho^{\otimes n}) = \sum_{\ve} \binom{n}{\ve} \vlambda^{\ve} \Tr_{n \to k} \parens*{\sigma(\ve)} = \sum_{\ve} \vlambda^{\ve} \frac{\Tr_{n \to k} \parens*{\sigma(\ve)}}{\Tr_{n \to 0}(\sigma(\ve))} = \sum_{\ve} \vlambda^{\ve} M_k^{\ve},
        \]
        and then dividing through by $Z$ gives the results for $M_1$ and $M_2$. 
    \end{proof}

    Since $Z^{\ve} = \Pr[\success \mid \ve] = \binom{n}{\ve}^{-1}$, and $Z = \Pr[\success]$ we can now see that the probability mass function for $\truedist$ is
    \[
    \Pr[\ve \mid \success] = \frac{\Pr[\success \mid \ve] \Pr[\ve]}{\Pr[\success]} = \frac{\vlambda^{\ve}}{Z}.
    \]
    In light of this, the expressions in the corollary appear to be expectations over $\ve \samp \truedist$, which is precisely what Theorem~\ref{thm:ms_from_mse} proves, so everything squares up nicely. It is a good time to also bound the probability the measurement is successful in terms of $\lambda_1$. 

    \begin{thm}
        \label{thm:probability_of_success}
        The probability of a successful measurement is $Z$ and 
        \begin{equation}
        \lambda_1^{n-1} \leq Z \leq \frac{\lambda_1^{n+1}}{2 \lambda_1 - 1} = \lambda_1^{n-1} (1 + \bigo{\eta^2}) \label{eq:meas_fail_prob}
        \end{equation}
    \end{thm}
    \begin{proof}
        Recall that $Z$ is the probability of success and $Z = \sum_{\ve} \vlambda^{\ve}$ over $\ve$ totaling $n$. On the one hand, it is lower bounded by the terms where $e_1 = n$ (i.e., $\lambda_1^{n}$) and $e_1 = n-1$ (i.e., $\lambda_1^{n-1} (\lambda_2 + \cdots + \lambda_d)$). It follows that 
        \[
        Z \geq \lambda_1^{n} + \lambda_1^{n-1}(\lambda_2 + \cdots + \lambda_d) = \lambda_1^{n-1}(\lambda_1 + \cdots + \lambda_d) = \lambda^{n-1}.
        \]
        On the other hand, since $\lambda_1$ is the dominant eigenvalue, it also makes sense to expand around it.
        \begin{align*}
            Z &= \sum_{e_1 + \cdots + e_d = n} \vlambda^{\ve} = \sum_{e_1=0}^{n} \sum_{e_2 + \cdots + e_d = n-e_1} \vlambda^{\ve}.
        \end{align*}
        We can then insert multinomial coefficients to simplify the $\lambda_2, \ldots, \lambda_d$ part.
        \begin{align*}
            Z &\leq \sum_{e_1=0}^{n} \sum_{e_2 + \cdots + e_d = n-e_1} \binom{n-e_1}{e_2,\ldots,e_d}\vlambda^{\ve} = \sum_{e_1=0}^{n} \lambda_1^{e_1} (\lambda_2 + \cdots + \lambda_d)^{n-e_1}
        \end{align*}
        Re-indexing and letting the sum extend to infinity, we have 
        \begin{align*}
            Z &\leq \sum_{j=0}^{n} \lambda_1^{n-j} (1 - \lambda_1)^{j}
            = \lambda_1^{n} \sum_{j=0}^{\infty} \parens*{\frac{1 - \lambda_1}{\lambda_1}}^{j}
            = \frac{\lambda_1^{n+1}}{2 \lambda_1 - 1}.
        \end{align*}
        Finally, $\eta = 1 - \lambda_1$ and we note that
        \[
            \frac{\lambda_1^2}{2\lambda_1 - 1} = \frac{(1- \eta)^{2}}{1 - 2\eta} = 1 + \eta^2 + \bigo{\eta^3},
        \]
        so the (multiplicative) gap between the two bounds is only $1 + \eta^2 + \bigo{\eta^3}$.
    \end{proof}
    
    \subsection{Geometric approximation}
    \label{sec:geometric_approximation}

    We now have expressions for the mean and variance of $\mhat$ in terms of $M_1$ and $M_2$ (Corollary~\ref{cor:estimator_mean_and_variance_wrt_ms}), expressions for $M_1$ and $M_2$ as expectations over $M_1^{\ve}$, $M_2^{\ve}$ (Theorem~\ref{thm:ms_from_mse}), expressions for $M_1^{\ve}$ and $M_2^{\ve}$ in terms of $\ve$ (Theorem~\ref{thm:sigma1}, Theorem~\ref{thm:sigma2}), and the distribution $\truedist$ for the expectation. There is one last obstacle to overcome: we would like to compute $\E_{\ve \in \truedist}[e_i]$ and $\E_{\ve \in \truedist}[e_i e_j]$, since those appear in $M_1$ and $M_2$. Exact expressions for these expectations have eluded us,\footnote{Also $Z$, which we could only upper and lower bound in Theorem~\ref{thm:probability_of_success}.} so we define an approximation, $\fakedist$, of the true distribution such that $\E_{\ve \in \fakedist}[e_i]$ is straightforward.

    Suppose the first eigenvalue is much larger than the rest, i.e., $\lambda_1 \gg \lambda_2 \geq \cdots \geq \lambda_d$. Hence, the $\ve$ vectors with highest probability in $\truedist$ have $e_1$ close to $n$, as large as possible. Let us rewrite the probability mass using the fact that $e_1 = n - e_2 - \cdots - e_d$. As long as $e_1, \ldots, e_d \in \mathbb N$, we have 
    \[
    f(\ve) = \frac{1}{Z} \lambda_1^{n-e_2-\ldots-e_d} \prod_{i=2}^{d} \lambda_i^{e_i} = \frac{\lambda_1^n}{Z} \prod_{i=2}^{d} \parens*{\frac{\lambda_i}{\lambda_1}}^{e_i}.
    \]
    It \emph{appears} that $\truedist$ factors as a product distribution on $e_2, \ldots, e_d$, i.e., it is proportional to $f_2(e_2) \cdots f_d(e_d)$ where $f_i(e_i) = (\frac{\lambda_i}{\lambda_1})^{e_i} (1 - \frac{\lambda_i}{\lambda_1})$ is the p.d.f.\ of a geometric random variable with mean $\frac{\lambda_i}{\lambda_1 - \lambda_i}$. We know $e_2, \ldots, e_d$ are not independent in $\truedist$, so there is a catch: in the \emph{very} unlikely event that $e_2 + \ldots + e_d$ exceeds $n$, the condition $e_1 = n - (e_2 + \cdots + e_d)$ requires us to set $e_1 < 0$. In fact, this is the only difference between the distributions.

    \begin{lemma}
        \label{lem:true_is_conditional_of_fake}
        The distribution $\truedist$ is exactly $\fakedist$ conditioned on $e_1 \geq 0$. 
    \end{lemma}
    \begin{proof}
        For $\ve$ with $e_1 \geq 0$ (the full support of $\truedist$) we have already seen that the p.d.f.\ $f$ factors as a product of $f_i$ (times a constant).
        \[
            f(\ve) = \frac{1}{Z} \lambda_1^{n-e_2-\ldots-e_d} \prod_{i=2}^{d} \lambda_i^{e_i} = \frac{\lambda_1^n}{Z} \prod_{i=2}^{d} \parens*{\frac{\lambda_i}{\lambda_1}}^{e_i} = \frac{\lambda_1^n}{Z} \prod_{i=2}^{d} \frac{f_i(e_i)}{1 - \tfrac{\lambda_i}{\lambda_1}}
        \]
        That is, whenever $e_1, \ldots, e_d \geq 0$, the two distributions are proportional. The only other $\ve$ with any support in $\fakedist$ are those with $e_1 < 0$, therefore if we condition on $e_1 \geq 0$ then $\fakedist$ becomes $\truedist$. 
    \end{proof}

    The two distributions are \emph{very} close to each other. We have consigned the proofs to Appendix~\ref{app:distributions}, but we quote the highlights below. First, the probability that $e_1 < 0$ is indeed very small, which in turn bounds the total variation distance, $\norm{\truedist - \fakedist}_{TV}$.
    \begin{restatable}{thm}{distributiontv}
        \label{thm:prob_e_negative}
        \[
        \Pr_{\ve \samp \fakedist}[e_1 < 0] \leq \Delta := \parens*{\frac{1 - \lambda_1}{\lambda_1}}^{n+1} \frac{\lambda_1}{2\lambda_1 - 1}.
        \]
        It follows that $\norm{\truedist - \fakedist}_{TV} = \Delta$.
    \end{restatable}
    When we rewrite $\Delta$ in terms of $\eta \leq \tfrac{1}{3}$,
    \[
    \Delta = \parens*{\frac{\eta}{1 - \eta}}^{n+1} \frac{1 - \eta}{1-2\eta} \leq 2 \cdot \parens*{\tfrac{3}{2}\eta}^{n+1}, 
    \]
    we see that $\Delta = \bigo{\eta^{2}}$, even if $n = 1$. More realistically, we will have $n \approx \frac{1}{\eta}$, and then $\Delta$ vanishes even more quickly, as shown in the plot Figure~\ref{fig:plot}. 
    \begin{figure}
    \centering
    \input{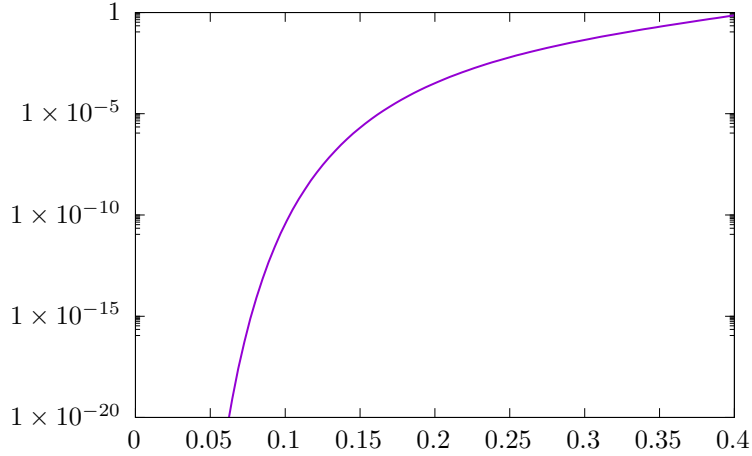}
    \caption{An upper bound on the mass of $\fakedist$ outside the support of $\truedist$ (as in Theorem~\ref{thm:prob_e_negative}) as a function of $\eta := 1 - \lambda_1$, assuming $n = 1/\eta$, semilog scale.} \label{fig:plot}
    \end{figure}

    We separately bound the change in $\E[\ve]$ for $\truedist$ versus $\fakedist$.
    \begin{restatable}{thm}{distributionmean}
        \label{thm:approx_first_order}
        The difference between $\truedist$ and $\fakedist$ for first-order expectations is at most 
        \[
        \norm{\E_{\truedist}[\ve] - \E_{\fakedist}[\ve]}_{1} = \sum_{i} \abs{\E_{\truedist}[e_i] - \E_{\fakedist}[e_i]} \leq \frac{2 \Delta}{1 - \Delta} \parens*{n + \frac{1}{2 \lambda_1 - 1}}.
        \]
    \end{restatable}
    Last, the variance difference is quantified with covariance matrices. 
    \begin{restatable}{lemma}{covariance} \label{lem:covariance_approx}
        The covariance matrices of $\truedist$ and $\fakedist$ are related as follows:
        \[
            \Sigma_{\fakedist} \succeq \Sigma_{\truedist} (1 - \Delta)^{2}.
        \]
    \end{restatable}

    \subsection{Mean of the estimator}
    \label{sec:mean_of_estimator}

    In this section, we find the mean of the estimator $\mhat$ conditioned on the success of the measurement (on state $\rho$). Recall that $\E[\mhat] = M_1$ by Corollary~\ref{cor:estimator_mean_and_variance_wrt_ms}, and Theorem~\ref{thm:ms_from_mse} expands this into 
    \[
        M_1 := M_1(\rho^{\otimes n}) =\E_{\ve \samp \truedist}[ M_{1}(\sigma(\ve))] = \E_{\ve \samp \truedist}[ M_1^{\ve} ].
    \]
    Below we approximate this expectation, except with the geometric random variable distribution $\fakedist$ in place of $\truedist$

    \begin{theorem}[Mean with $\fakedist$]
        \label{thm:estimator_mean_fakedist}
        \[
            \E_{\ve \in \fakedist}[M_1^{\ve}] = \E_{\ve \in \fakedist} \bracks*{\frac{1}{n} \parens*{\sum_{i=1}^{d} e_i \Phi_i}} =
            \Phi_1 + \frac{1}{n} \sum_{j=2}^{d} \tfrac{\lambda_j}{\lambda_1 - \lambda_j} \parens*{ \Phi_j - \Phi_1 }
        \]
    \end{theorem}
    \begin{proof}
        We rewrite with $e_1 = n - e_2 - \cdots - e_d$ and use that $\E_{\ve \in \fakedist}[e_i] = \frac{\lambda_i}{\lambda_1 - \lambda_i}$ for all $2 \leq i \leq d$ (by Proposition~\ref{prop:geometric_moments}).
        \begin{align*}
            \E_{\ve \in \fakedist} \bracks*{\frac{1}{n} \parens*{\sum_{i=1}^{d} e_i \Phi_i}} 
            &= \frac{1}{n} \parens*{n \Phi_1 + \sum_{i=2}^{d} \E_{\ve \in \fakedist}[e_i] (\Phi_i - \Phi_1)}
            = \Phi_1 + \frac{1}{n} \sum_{i=2}^{d} \frac{\lambda_i}{\lambda_1 - \lambda_i} (\Phi_i - \Phi_1).
        \end{align*}
    \end{proof}
    \begin{cor}
        \label{cor:estimator_mean_truedist}
        \[
        \tracenorm{ M_1 - \Phi_1 - \frac{1}{n} \sum_{j=2}^{d} \tfrac{\lambda_j}{\lambda_1 - \lambda_j} \parens*{ \Phi_j - \Phi_1 }} = \bigo{\eta^2} %\leq \qboxed{\Delta \frac{(n+1)^2}{(2 \lambda_1 - 1)^2}}
        \]
    \end{cor}
    \begin{proof}
        This result is about comparing $M_1 = \E_{\ve \in \truedist}\bracks*{M_1^{\ve}}$ with 
        \[
        \E_{\ve \in \fakedist}[M_1^{\ve}] = \E_{\ve \in \fakedist}\bracks*{\frac{1}{n} \parens*{\sum_{i=1}^{d} e_i \Phi_i}} = \Phi_1 + \frac{1}{n} \sum_{j=2}^{d} \tfrac{\lambda_j}{\lambda_1 - \lambda_j} \parens*{ \Phi_j - \Phi_1 }
        \]
        from the previous theorem. That is, we are bounding the difference in $M_1$ due to our approximation of $\truedist$ with $\fakedist$. Since the operators are diagonal in the $\Phi_i$ basis, the trace norm simplifies:
        \[
        \frac{1}{n} \parens*{\tracenorm{\sum_{i=1}^{d} (\E_{\ve \in \truedist}\bracks*{e_i} - \E_{\ve \in \fakedist}\bracks*{e_i}) \Phi_i}} = \frac{1}{n} \sum_{i=1}^{d} \abs{\E_{\ve \in \truedist}[e_i] - \E_{\ve \in \fakedist}[e_i]}.
        \]
        Theorem~\ref{thm:approx_first_order} upper bounds this by 
        \[
        \frac{1}{n} \frac{2 \Delta}{1 - \Delta} \parens*{n + \frac{1}{2 \lambda_1 - 1}} = \bigo{\Delta} \subseteq \bigo{\eta^2}. 
        \]
    \end{proof}

    \subsection{Variance of the estimator}
    \label{sec:variance_of_estimator}

    Recall that there are two sources of variance for $\Tr(O \mhat)$: mixture randomness (from $\rho^{\otimes n}$ being a mixture of pure states $\sigma(\ve)$), and the inherent randomness of quantum measurement. These two sources of randomness are responsible for variance of $\Tr(O \mhat)$, as formalized by the law of total variance.
    \begin{thm}[Law of total variance]
		\label{thm:total_variance}
		\[
			\Var_{\ve \samp \truedist,\meas}(\Tr(O \mhat)) = \Var_{\ve \samp \truedist}[ \E_{\meas}(\Tr(O\mhat) \mid \ve)] + \E_{\ve \samp \truedist}[ \Var_{\meas}(\Tr(O\mhat) \mid \ve)].
		\]
	\end{thm}
    \noindent We go on to bound these terms individually. Since the expressions for variance get somewhat unwieldy, we introduce shorthand common terms involving $O$ and $\Phi_i$: let $O_i := \Tr(O \Phi_i)$ and $O_{ij} := \Tr(O \Phi_i O \Phi_j)$. In general $O_i O_j \neq O_{ij}$, but it will be important that
    \[
    O_{kk} = \Tr(O \Phi_k O \Phi_k) = \Tr(O \Phi_k)^2 = O_k^2
    \]
    for all $1 \leq k \leq d$. We also introduce $\mhat^{\ve}$ to represent the estimator conditioned on $\sigma(\ve)$ being the input state. 

    \subsubsection{Variance due to mixture randomness}
    
	\begin{thm}
		\label{thm:mixture_variance}
		The variance in $\Tr(O \mhat)$ due to $\rho^{\otimes n}$ being a mixture of $\sigma(\ve)$ is 
		\[
		\Var_{\ve \in \truedist}[\E_{\textup{meas}}(\Tr(O\mhat) \mid \ve)] = \Var_{\ve \in \truedist}[\E_{\textup{meas}}(\Tr(O\mhat^{\ve})] \leq \frac{4 \norm{O}_\infty^2}{n^2 (1 - \Delta)^2}\frac{\lambda_1(1-\lambda_1)}{(2 \lambda_1 - 1)^2}.
		\]
	\end{thm}
    \begin{proof}
        First, $\E_{\meas}(\Tr(O \mhat^{\ve}) \mid \ve) = \Tr(O \E_{\meas}(\mhat^{\ve})) = \Tr(O M_1^{\ve})$ by Corollary~\ref{cor:estimator_mean_and_variance_wrt_ms} --- the corollary is for a general state $A$, so it applies with $\sigma(\ve)$. 

        Theorem~\ref{thm:sigma1} gives an expression for $M_1^{\ve}$:
        \[
            \Tr(O M_1^{\ve}) = \frac{1}{n} \sum_{i=1}^{d} e_i \Tr(O \Phi_i) = \frac{1}{n} \sum_{i=1}^{d} e_i O_i = \frac{\vv^{\top} \ve}{n}, %= O_1 + \frac{1}{n} \sum_{i=2}^{d} e_i (O_i - O_1). 
        \]
        where $\vv = (O_1, \ldots, O_d)$ is the vector of $O_i$s. The variance is thus 
        \[
        \Var_{\ve \samp \truedist}[\Tr(O M_1^{\ve})] = \frac{1}{n^2} \Var_{\ve \samp \truedist}[\vv^{\top} \ve] = \frac{\vv^{\top} \Sigma_{\truedist} \vv}{n^2},
        \]
        where $\Sigma_{\truedist}$ is the covariance matrix for $\truedist$. By Lemma~\ref{lem:covariance_approx}, $\Sigma_{\truedist} \preceq \frac{1}{(1 - \Delta)^2} \Sigma_{\fakedist}$, and thus 
        \[
        \Var_{\ve \samp \truedist}[\Tr(O M_1^{\ve})] = \frac{\vv^{\top} \Sigma_{\truedist} \vv}{n^2} \leq \frac{\vv^{\top} \Sigma_{\fakedist} \vv}{n^2 (1 - \Delta)^2} = \frac{1}{(1 - \Delta)^2} \Var_{\ve \samp \fakedist}\bracks*{\frac{1}{n} \sum_{i=1}^{d} e_i O_i}.
        \]
        Since $e_2, \ldots, e_d$ are independent under $\fakedist$, we rewrite with $e_1 = n - e_2 - \cdots - e_d$ and simplify as much as possible.
        \[
        \Var_{\ve \samp \fakedist}\bracks*{\frac{1}{n}\sum_{i=1}^{d} e_i O_i} = \Var_{\ve \samp \fakedist}\bracks*{O_1 + \frac{1}{n} \sum_{i=2}^{d} e_i (O_i - O_1)} = \frac{1}{n^2} \sum_{i=2}^{d}\Var_{\ve \samp \fakedist}\bracks*{e_i} (O_i - O_1)^2
        \]
        Since $(O_i - O_1)^2 \leq 4\inftynorm{O}^2$ and $\Var_{\ve \samp \fakedist}[e_i] \leq \frac{\lambda_1 \lambda_i}{\lambda_1 - \lambda_i} \leq \frac{\lambda_1 \lambda_i}{2 \lambda_1 - 1}$, the variance is bounded by 
        \[
        \Var_{\ve \samp \truedist}[\Tr(O M_1^{\ve})] \leq \frac{1}{n^2(1- \Delta)^2} \sum_{i=2}^{d}\Var_{\ve \samp \fakedist}\bracks*{e_i} (O_i - O_1)^2 \leq \frac{4 \inftynorm{O}^2}{n^2 (1 - \Delta)^2} \frac{\lambda_1 (1 - \lambda_1)}{(2 \lambda_1 - 1)^2}.
        \]
    \end{proof}
    
    We note that the variance due to mixture randomness is small in all the ways we want: it is a function of $\norm{O}^2$ rather than $\Tr(O^2)$, it is quadratic (rather than linear) in $\frac{1}{n}$, and it is multiplied by a factor of $\eta = 1 - \lambda_1$. We proceed with the analysis of the other (dominant) term in the variance. 
    
    \subsubsection{Variance due to measurement randomness}

	\begin{lemma}
		\label{lem:meas_variance_part1}
		\begin{equation}
    		\Var_{\meas}[\Tr(O \mhat) \mid \ve] \leq \frac{\Tr(O^2)}{n^2} + \frac{2\inftynorm{O}^2}{n} + \frac{1}{n^2} \sum_{i \neq j} e_i e_j O_{ij} \label{eq:measurevar_start}
        \end{equation}
	\end{lemma}
    \begin{proof} 
        Recall that Corollary~\ref{cor:estimator_mean_and_variance_wrt_ms} already gives 
        \[
        \Var_{\meas}[\Tr(O\mhat^{\ve})] \leq \frac{\Tr(O^2)}{n^2} + \frac{2 \inftynorm{O}^2}{n} + \frac{n-1}{n} \Tr(O^{\otimes 2} M_2^{\ve}) - \Tr(O M_1^{\ve})^2,
        \]
        by applying it to $\sigma(\ve)$. The first two terms match the goal, so we focus on bounding the last two terms, using the expressions for $M_1^{\ve}$ and $M_2^{\ve}$ from Theorem~\ref{thm:sigma1} and Theorem~\ref{thm:sigma2}.

        \begin{align*}
            \qquad n(n-1) \Tr(O^{\otimes 2} M_2^{\ve})
            &= \Tr \parens*{ O^{\otimes 2} 2\sym^{(2)} \parens*{\sum_{i \neq j} e_i e_j \Phi_i \otimes \Phi_j + \sum_{k} \binom{e_k}{2} \Phi_k^{\otimes 2}} } \\
            &= \sum_{i \neq j} e_i e_j (O_i O_j + O_{ij})+ \tfrac{1}{2} \sum_{k} (e_k^2 - e_k) (O_k^2 + O_{kk}) \\
            &= \sum_{i \neq j} e_i e_j O_i O_j + \sum_{i \neq j} e_i e_j O_{ij} + \sum_{k} e_k^2 O_k^2 - \sum_{k} e_k O_k^2 \\
            &= \sum_{i, j} e_i e_j O_i O_j + \sum_{i \neq j} e_i e_j O_{ij} - \sum_{k} e_k O_k^2 \\
            &= n^2 \Tr(O M_1^{\ve})^{2} + \sum_{i \neq j} e_i e_j O_{ij} - \sum_{k} e_k O_k^2
        \end{align*}
        We can drop the negative term, and then it follows that 
        \[
            \frac{n-1}{n} \Tr(O^{\otimes 2} M_2^{\ve}) - \Tr(O M_1^{\ve})^2 \leq \sum_{i \neq j} e_i e_j O_{ij},
        \]
        from which we get the result.
    \end{proof}

    Now let us separately bound the last term of \eqref{eq:measurevar_start}.
    \begin{lemma}
        \label{lem:meas_variance_part2}
        \begin{align*}
            \frac{1}{n^2} \E_{\ve \samp \truedist}\bracks*{\sum_{i \neq j} e_i e_j O_{ij}} &\leq \frac{2}{n} \frac{1 - \lambda_1}{2 \lambda_1 - 1} \norm{O}_{\infty}^{2} + \bigo{\Delta}
        \end{align*}
    \end{lemma}
    \begin{proof}
        First, observe that 
        \begin{align*}
        O_{ij} &= \Tr(O \Phi_i O \Phi_j) = \langle \phi_i | O | \phi_j \rangle \langle \phi_j | O | \phi_i \rangle = \abs{\langle \phi_i | O | \phi_j \rangle}^2 \leq \norm{O}_{\infty}^{2}.
        \end{align*}
        In other words, we can bound each $O_{ij}$ by $\norm{O}_{\infty}^{2}$, and then it clearly suffices to bound $\sum_{i \neq j} e_i e_j$. To start, we can look at the sum as being over all $e_i e_j$ -- which totals $n^2$ on the basis that $e_1 + \cdots + e_d = n$ for all $\ve$ -- minus the ``diagonal'' terms, of which we claim only $e_1^2$ will be relevant.
        \begin{align*}
            \sum_{i \neq j} \E_{\ve \in \truedist}[e_i e_j] 
            &= \sum_{i, j} \E_{\ve \in \truedist}[e_i e_j] - \sum_{k} \E_{\ve \in \truedist}[e_k^2] \\
            &= \E_{\ve \in \truedist}\bracks*{\sum_{i, j} e_i e_j} - \sum_{k} \E_{\ve \in \truedist}[e_k^2] \\
            &\leq n^2 - \E_{\ve \samp \truedist}[e_1^2]
        \end{align*}
        Write $e_1$ as $n - e_2 - \cdots - e_d$, and we get 
        \begin{align*}
            \E_{\ve \samp \truedist}[e_1^2] 
            &= \E_{\ve \samp \truedist}\bracks*{\parens*{n - \sum_{i=2}^{d} e_i}^2} \\
            &= n^2 - 2n \sum_{i=2}^{d} \E_{\ve \samp \truedist}\bracks*{e_i} + \sum_{i=2}^{d} \sum_{j=2}^{d}\E_{\ve \samp \truedist}\bracks*{e_i e_j} \\
            &\geq n^2 - 2n \sum_{i=2}^{d} \E_{\ve \samp \truedist}[e_i]. 
        \end{align*}
        It follows that $\sum_{i \neq j} \E_{\ve \samp \truedist}[e_i e_j] \leq 2n \sum_{i=2}^{d} \E_{\ve \samp \truedist}[e_i]$. Under the approximate distribution, this is 
        \[
        \sum_{i=2}^{d} \E_{\ve \samp \fakedist}[e_i] = \sum_{i=2}^{d} \frac{\lambda_i}{\lambda_1 - \lambda_i} \leq \frac{1 - \lambda_1}{2\lambda_1 - 1},
        \]
        and Theorem~\ref{thm:approx_first_order} bounds the difference from the true distribution by at most $\bigo{\Delta}$. The result follows.
    \end{proof}
    \subsection{Conclusion}

    We finish the section by stating and proving a more formal version of Theorem~\ref{thm:robustness_of_joint_measurement}.
    \begin{theorem}
        For unknown state with deviation $\eta$, the standard joint measurement on $n$ copies succeeds with probability at least $(1-\eta)^{n-1}$. Conditioned on success, there is an estimator $\phihat$ such that 
        \begin{align*}
            \E[\Tr(O\phihat)] &= \Tr(O\Phi_1) + \frac{1}{n} \cdot \parens*{\frac{\Tr(O\rho) - \Tr(O\Phi_1)}{1 - \eta} + \bigo{\inftynorm{O}\eta^2}} \\
            \Var[\Tr(O\phihat)] &= \frac{\Tr(O^2)}{n^2} + \frac{6\inftynorm{O}^2}{n} + \frac{8 \inftynorm{O}^2}{n^2} + \bigo{\Delta}
        \end{align*}
    \end{theorem}
    \begin{proof}
        Take $\phihat = \mhat$.  Theorem~\ref{thm:estimator_mean_fakedist} gives the mean under the approximate distribution
        \[
            \E_{\ve \samp \fakedist}[M_1^{\ve}] =
            \Phi_1 + \frac{1}{n} \sum_{j=2}^{d} \tfrac{\lambda_j}{\lambda_1 - \lambda_j} \parens*{ \Phi_j - \Phi_1 }.
        \]
        First, expand $\frac{\lambda_j}{\lambda_1 - \lambda_j}$ as:
        \[
            \frac{\lambda_j}{\lambda_1 - \lambda_j} = \frac{\lambda_j}{\lambda_1}\parens*{1 + \bigo{\frac{\lambda_j}{\lambda_1}}} = \frac{\lambda_j}{\lambda_1} \parens*{1 + \bigo{\frac{\eta}{1 - \eta}}} = \frac{\lambda_j}{\lambda_1} \parens*{1 + \bigo{\eta}}.
        \]
        The first order term gives
        \[
        \sum_{j=2}^{d} \frac{\lambda_j}{\lambda_1} (\Phi_j - \Phi_1) = \sum_{j=1}^{d} \frac{\lambda_j}{\lambda_1} (\Phi_j - \Phi_1) = \rho - \Phi_1.
        \]
        The second order term is the same, but multiplied by $\bigo{\eta}$, and since 
        \[
        \Tr(O(\rho - \Phi_1)) \leq \inftynorm{O} \tracenorm{\rho - \Phi_1} = \bigo{\inftynorm{O} \eta},
        \]
        the expectation is indeed
        \[
            \E_{\ve \samp \fakedist}[\Tr(OM_1^{\ve})] =
            \Tr(O\Phi_1) + \frac{1}{n} \parens*{\frac{\Tr(O \rho) - \Tr(O \Phi_1)}{1 - \eta} + \bigo{\inftynorm{O} \eta^2}}.
        \]
        This is the expectation under $\fakedist$, but Corollary~\ref{cor:estimator_mean_truedist} proves the approximation changes the trace distance by at most $\bigo{\Delta} \subseteq \bigo{\eta^2}$, and thus affects the final expectation by $\bigo{\inftynorm{O} \eta^2}$.

        On the variance side,     Theorem~\ref{thm:total_variance} divides the variance into a sum of mixture randomness and measurement randomness. Theorem~\ref{thm:mixture_variance} bounds the mixture randomness:
        \[
            \Var_{\ve \in \truedist}[\E_{\textup{meas}}(\Tr(O\mhat) \mid \ve)] \leq \frac{4 \norm{O}^2}{n^2}\frac{\lambda_1(1-\lambda_1)}{(2 \lambda_1 - 1)^2} + \bigo{\Delta}.
        \]
        A combination of Lemma~\ref{lem:meas_variance_part1} and Lemma~\ref{lem:meas_variance_part2} bounds the measurement randomness:
        \[ 
    		\Var_{\meas}[\Tr(O \mhat) \mid \ve] \leq \frac{\Tr(O^2)}{n^2} + \frac{2\inftynorm{O}^2}{n} + \frac{2}{n} \frac{1 - \lambda_1}{2 \lambda_1 - 1} \norm{O}_{\infty}^{2} + \bigo{\Delta}.
        \]
        The total is 
        \begin{align*}
            \Var[\tr{O\rhohat}^2] 
            &\leq \frac{\Tr(O^2)}{n^2} + \frac{2\inftynorm{O}^2}{n} + \frac{2\inftynorm{O}^{2}}{n} \frac{1 - \lambda_1}{2 \lambda_1 - 1} + \frac{4 \inftynorm{O}^2}{n^2}\frac{\lambda_1(1-\lambda_1)}{(2 \lambda_1 - 1)^2} + \bigo{\Delta} \\
            &\leq \frac{\Tr(O^2)}{n^2} + \frac{6\inftynorm{O}^2}{n} + \frac{8 \inftynorm{O}^2}{n^2} + \bigo{\Delta}.
        \end{align*}
    \end{proof}

\section{Chiribella}
\label{app:chiribella}

    In this section we adapt Chiribella's theorem \cite{chiribella2011quantum} to get expressions for $\E[\Psi]$ and $\Var[\Psi]$. The subject of this theorem is the map $\MP_{n \to k}$ defined below. 
    \begin{defn}
        \label{defn:measure_and_prepare}
        For integers $n, k \geq 0$, let $\MP_{n \to k} \colon \link{n} \to \link{k}$ be such that
        \[
        \MP_{n \to k}(A) = \frac{\symdim{d}{n}}{\symdim{d}{n+k}} \Tr_{[n]} \parens{\sym^{(n+k)}(A \otimes \eye^{\otimes k})} = \frac{\symdim{d}{n}}{\symdim{d}{n+k}} \Tr_{n+k \to k}(A \otimes I^{\otimes k})
        \]
        for all $A \in \link{n}$. We remind the reader that $\Tr_{[n]}$ is the partial trace over qudits $[n] = \{ 1, \ldots, n \}$.
    \end{defn}
    This map is an example of a ``measure and prepare map'' because it is equivalent to measuring the state with some POVM, and then preparing a state dependent on the outcome. In particular, Proposition~\ref{prop:moments_and_measp} below shows that this map measures with $\mathcal{M}_n$, and prepares $\ketbra{\psi}{\psi}^{\otimes k}$ if the outcome is $\psi$, or $0$ if the measurement fails.
    \begin{prop}
        \label{prop:moments_and_measp}
        Let $\Psi$ be the outcome of measuring an $n$-qudit state $A$ with $\mathcal{M}_s$ ($\ketbra{\psi}{\psi}$ or $0$ for failure). Then the $k^{\textrm{th}}$ moment of $\Psi$ is $\E[\Psi^{\otimes k}] = \MP_{n \to k}(A)$
        for all $k \geq 0$. %Moreover, $\E[\Psi^{\otimes k} \mid \success] = \MP_{n \to k}(A) / \Tr_{n \to 0}(A)$.
    \end{prop}
    \begin{proof}
        The expectation of $\Psi^{\otimes k}$ is a straightforward calculation using definition of $\mathcal{M}_n$ (Definition~\ref{defn:measurement}) and the Haar integral characterization of $\sym^{(n)}$ (Lemma~\ref{lem:symmetric_haar}) as needed.
        \begin{align*}
            \E[\Psi^{\otimes k}] &= 0^{\otimes k} \Tr(F_{\fail} A) + \int \ketbra{\psi}{\psi}^{\otimes k} \Tr(F_{\psi} A) \\
            &= \int \ketbra{\psi}{\psi}^{\otimes k} \symdim{d}{n} \Tr(\ketbra{\psi}{\psi}^{\otimes n} A) \mathrm{d} \psi & \text{definition of $F_{\psi}$} \\
            &= \symdim{d}{n} \Tr_{[n]} \parens*{ \parens*{\int\ketbra{\psi}{\psi}^{\otimes (n+k)}  \mathrm{d} \psi}(A \otimes \eye^{\otimes k})} & \text{linearity of trace, integral} \\
            &= \frac{\symdim{d}{n}}{\symdim{d}{n+k}} \Tr_{[n]} \parens*{\sym^{(n+k)} (A \otimes \eye^{\otimes k})} \\
            &= \MP_{n \to k}(A). & \text{definition of $\MP$}
        \end{align*}
    \end{proof}

    In addition to $\MP_{n \to k}$, Chiribella's theorem uses a ``cloning map'', defined below. 
    \begin{defn}[Optimal Cloning Map \cite{werner1998optimal}]
        \label{defn:clone_map}
        Let us define the superoperator $\clone_{n \to n+k} \colon \link{n} \to \link{n+k}$ on input $A \in \link{n}$ as 
        \[
            \clone_{n \to n+k}(A) = \frac{\symdim{d}{n}}{\symdim{d}{n+k}} \sym^{(n+k)} \parens*{A \otimes \eye^{\otimes k}} \sym^{(n+k)}.
        \]
    \end{defn}
    This map extends an $n$-qubit state to $n+k$ qudits. The no-cloning theorem prohibits cloning quantum states, but Werner \cite{werner1998optimal} showed that it is the optimal with respect to the fidelity of $\clone_{n \to n+k}(\sigma^{\otimes n})$ and $\sigma^{\otimes n+k}$. 

    This brings us to the key result of this section, due to \cite{chiribella2011quantum}.
    \begin{theorem}[Chiribella's theorem]
        \label{thm:chiribella}
        For $A \in \Pi_{\mathrm{sym}}^{(n)} ((\mathbb C^d)^{\otimes n})$ (in the symmetric subspace)
        \[
            \MP_{n \to k}(A) = \binom{d+k+n-1}{k}^{-1} \sum_{s=0}^{k} \binom{n}{s} \binom{d+k-1}{k-s} \clone_{s \to k} \parens*{\Tr_{[n-s]}(A)}.
        \]
    \end{theorem}
    Before we get into the relevance of this theorem, let us quickly upgrade it from symmetric states to exchangeable states.
    \begin{cor}
        \label{cor:chiribella_prime}
        For exchangeable $A \in \link{n}$, 
        \[
            \MP_{n \to k}(A) = \binom{d+k+n-1}{k}^{-1} \sum_{s=0}^{k} \binom{n}{s} \binom{d+k-1}{k-s} \clone_{s \to k} \parens*{\Tr_{n \to s}(A)},
        \]
        where $\Tr_{n \to s}(A) = \Tr_{[n-s]}(\sym^{(n)} A)$ is taken from Section~\ref{sec:preliminaries}.
    \end{cor}
    \begin{proof}
        Set $A' = \sym^{(n)} A$ and observe that $A'$ is symmetric, since $\sym^{(n)}$ clearly absorbs permutations on the left, and 
        \[
        A' W_{\pi} = \sym^{(n)} A W_{\pi} = \sym^{(n)} W_{\pi}^{\dag} A = \sym^{(n)} A = A'
        \]
        on the left, using the exchangeability of $A$. It follows that we can plug $A'$ into Theorem~\ref{thm:chiribella}, to get that the right hand side of the claim equals $\MP_{n \to k}(A')$. Then we expand with Definition~\ref{defn:measure_and_prepare} and see that the extra $\sym^{(n)}$ can be absorbed into $\sym^{(n+k)}$:
        \[
        \MP_{n \to k}(A') = \frac{\symdim{d}{n}}{\symdim{d}{n+k}} \Tr_{[n]} \parens{\sym^{(n+k)}(\sym^{(n)}A \otimes \eye^{\otimes k})} = \frac{\symdim{d}{n}}{\symdim{d}{n+k}} \Tr_{[n]} \parens{\sym^{(n+k)}(A \otimes \eye^{\otimes k})} = \MP_{n \to k}(A).
        \]
        The result follows.
    \end{proof}
    The relevance of Chiribella's theorem (or the corollary) is that it expresses $\MP_{n \to k}(A)$, and thus $\E[\Psi^{\otimes k}]$ (by Proposition~\ref{prop:moments_and_measp}), in terms of a handful of partial traces. It distills the state $A$ down to $\leq k$ qudits by partial trace, then blows it back up to $k$ qudits with the cloning map. In other words, we can compute $\E[\Psi]$ and $\E[\Psi^{\otimes 2}]$ entirely from $1$-qudit and $2$-qudit summaries (i.e., $\Tr_{[n-1]}(\sym^{(n)} A)$ and $\Tr_{[n-2]}(\sym^{(n)} A)$) of the full $n$-qudit state $A$. 

    We refactor Chiribella one more time to (i) explicitly link the calculation to $\E[\Psi^{\otimes k}]$, (ii) simplify the binomial coefficients as much as possible, and (iii) expand $\clone_{k \to s}$ with its definition so that it is not needed in the main text. This is the version of Chiribella's theorem we quote in Section~\ref{sec:chiribella}.

    \chiribella*
    \begin{proof}
        Start from the definition of $\MP$, plug in the definition of $\clone$, and expand the binomials to simplify.
        \begin{align*}
            \MP_{n \to k}(A) &= \binom{d+k+n-1}{k}^{-1} \sum_{s=0}^{k} \binom{n}{s} \binom{d+k-1}{k-s} \clone_{s \to k}(\Tr_{n \to s}(A)) \\
            &= \binom{d+k+n-1}{k}^{-1} \sym^{(k)} \parens*{\sum_{s=0}^{k} \binom{n}{s} \binom{d+k-1}{k-s} \frac{\symdim{d}{s}}{\symdim{d}{k}}  \parens*{\Tr_{n \to s}(A) \otimes \eye^{\otimes k-s}} }\sym^{(k)} \\
            &= \frac{1}{\symdim{(d+n)}{k}} \sym^{(k)} \parens*{\sum_{s=0}^{k} \binom{n}{s} \frac{(d+k-1)!}{(k-s)!(d+s-1)!} \frac{\frac{(d+s-1)!}{(d-1)!s!}}{\frac{(d+k-1)!}{(d-1)!k!}}  \parens*{\Tr_{n \to s}(A) \otimes \eye^{\otimes k-s}} }\sym^{(k)} \\
            &= \frac{1}{\symdim{(d+n)}{k}} \sym^{(k)} \parens*{\sum_{s=0}^{k} \binom{n}{s} \binom{k}{s} \parens*{\Tr_{n \to s}(A) \otimes \eye^{\otimes k-s}} }\sym^{(k)}.
        \end{align*}
        Recall that $\E[\Psi^{\otimes k}] = \MP_{n \to k}(A)$ for exchangeable $A$ by Theorem~\ref{cor:chiribella_prime}. Since 
        \begin{align*}
        \E[\Psi^{\otimes k}] &= \E[\Psi^{\otimes k} \mid \success] \Pr[\success] + \E[\Psi^{\otimes k} \mid \neg \success] \Pr[\neg \success] \\
        &= \E[\Psi^{\otimes k} \mid \success] \Tr(\sym^{(n)}A),
        \end{align*}
        we can divide through by $\Tr_{n \to 0}(A) := \Tr(\sym^{(n)} A)$ to get the expectation of $\Psi^{\otimes k}$ conditioned on success. 
    \end{proof}

    \section{Approximating the distribution of \texorpdfstring{$\ve$}{e}}
    \label{app:distributions}

    This appendix is dedicated to results about the approximate distribution $\fakedist$, and how it relates to $\truedist$ and the expectation values we wish to compute .

    First, we recall the mean and variance of the geometric random variables composing $\ve \samp \fakedist$.
    \begin{prop}
        \label{prop:geometric_moments}
        Let $\ve \samp \fakedist$. The mean and variance of $e_i$ for $2 \leq i \leq d$ is 
        \begin{align*}
            \E[e_i] &= \frac{\lambda_i}{\lambda_1 - \lambda_i}, & 
            \Var[e_i] &= \frac{\lambda_i \lambda_1}{(\lambda_1 - \lambda_i)^2},
        \end{align*}
    \end{prop}
    
    Next, we want to bound the difference between $\truedist$ and $\fakedist$. To start, we bound the probability $e_i = n - j$ for arbitrary $j$. 
    \begin{lemma}
        \label{lem:prob_e_equals}
        When $\ve \sim \fakedist$, we have $
        \Pr[e_2 + \cdots + e_d = j = n - e_1] \leq \parens*{\frac{1-\lambda_1}{\lambda_1}}^j
        $.
    \end{lemma}
    \begin{proof}
        From the definitions, we have
        \[
            \Pr[e_2 + \cdots + e_d = j] 
            = \sum_{e_2 + \cdots + e_d = j} \prod_{i=2}^{d} f_i(e_i)
            = \sum_{e_2 + \cdots + e_d = j} \prod_{i=2}^{d} (\tfrac{\lambda_i}{\lambda_1})^{e_i}(1 - \tfrac{\lambda_i}{\lambda_1}).
        \]
        The factors $(1 - \tfrac{\lambda_i}{\lambda_1})$ are all $\leq 1$ and can be neglected. Then we introduce multinomial coefficients $\binom{j}{\ve} \geq 1$ into the sum, letting us apply the multinomial theorem.
        \begin{align*}
            \Pr[e_2 + \cdots + e_d = j]  &\leq \sum_{e_2 + \cdots + e_d = j} \binom{j}{\ve} \prod_{i=2}^{d} (\tfrac{\lambda_i}{\lambda_1})^{e_i} = \parens*{\sum_{i=2}^{d} \frac{\lambda_i}{\lambda_1}}^{j} = \parens*{\frac{1 - \lambda_1}{\lambda_1}}^{j}.
        \end{align*}
    \end{proof}
    This leads to a bound on the probability $e_1$ is negative, which then also bounds the distance between the two distributions. 
    \distributiontv*
    \begin{proof}
        Recall that $e_1 = n - e_2 - \cdots - e_d$, so 
        \begin{align*}
            \Pr_{\ve \sim \fakedist}[e_1 < 0] = \Pr_{\ve \sim \fakedist}[e_2 + \cdots + e_d > n] = \sum_{j=n+1}^{\infty} \Pr_{\ve \sim \fakedist}[e_2 + \cdots + e_d = j].
        \end{align*}
        Use Lemma~\ref{lem:prob_e_equals} and sum a geometric series to get 
        \begin{align*}
            \Pr_{\ve \sim \fakedist}[e_1 < 0] 
            &\leq \sum_{j=n+1}^{\infty} \parens*{\frac{1-\lambda_1}{\lambda_1}}^{j}= \parens*{\frac{1 - \lambda_1}{\lambda_1}}^{n+1} \frac{\lambda_1}{2\lambda_1 - 1} = \Delta.
        \end{align*}
        We know that where $\truedist$ has support, the mass is proportional to $\fakedist$, but necessarily larger because $\fakedist$ has mass on $e_1 < 0$ where $\truedist$ does not. Moreover, $e_1 < 0$ is the only area where $\fakedist$ has support and $\truedist$ does not. Thus, $\Pr[e_1 < 0]$ is exactly the mass which must be moved to transform $\fakedist$ into $\truedist$, and hence $\norm{\truedist - \fakedist}_{TV} = \Delta$.
    \end{proof}
    We discuss in the main text that $\Delta = \bigo{\eta^2}$, and can be \emph{very} small for practical values of $\lambda_1$ and $n$.

    Recall that we need $\E_{\ve \samp \truedist}[e_i]$ and $\E_{\ve \samp \truedist}[e_i e_j]$ to evaluate $M_1$ and $M_2$.  The total variation distance alone is insufficient to bound the difference in expectation, so we must separately justify how much the approximation can distort expectations. 
    \distributionmean*
    \begin{proof}
        For some small probability $p := \Pr_{\fakedist}[e_1 < 0] \leq \Delta$, we have 
        \begin{align*}
            \E_{\fakedist}[\ve] &= \Pr[e_1 \geq 0] \E_{\fakedist}[\ve \mid e_1 \geq 0] + \Pr[e_1 < 0] \E_{\fakedist}[\ve \mid e_1 < 0] \\
            &= (1-p)\E_{\truedist}[\ve] + p\E_{\fakedist}[\ve \mid e_1 < 0].
        \end{align*}
        Rearranging, $\E_{\truedist}[\ve] = \frac{1}{1-p} \parens*{\E_{\fakedist}[\ve] - p \E_{\fakedist}[\ve \mid e_1 < 0]}$, and thus  
        \[
        \norm{\E_{\truedist}[\ve] - \E_{\fakedist}[\ve]}_{1} = \frac{p}{1-p} \norm{\E_{\fakedist}[\ve] - \E_{\fakedist}[\ve \mid e_1 < 0]}_{1}.
        \]
        Since the sum of $\ve$ is always $n$ under any of these distributions, the expectations of $\ve$ also sum to $n$. The sum of the coordinate-wise differences is $n - n = 0$, and thus the absolute difference on $e_1$ is, by triangle inequality, bounded by the absolute differences on for the other coordinates. That is,
        \[
        \norm{\E_{\fakedist}[\ve] - \E_{\fakedist}[\ve \mid e_1 < 0]}_{1} \leq 2\norm{\E_{\fakedist}[\ve_{-1}] - \E_{\fakedist}[\ve_{-1} \mid e_1 < 0]}_{1} \leq 2\norm{\E_{\fakedist}[\ve_{-1}]}_1 + 2\norm{\E_{\fakedist}[\ve_{-1} \mid e_1 < 0]}_1,
        \]
        where $\ve_{-1} = (e_2, \ldots, e_d)$.

        Since $e_2, \ldots, e_d \geq 0$, these norms are both just the sum of the entries, i.e.,
        \[
        \norm{\E_{\fakedist}[\ve_{-1}]}_1 = \sum_{i=2}^{d} \E_{\fakedist}[e_i].
        \]
        For these geometric random variables, we have $\E[e_i] = \frac{\lambda_i}{\lambda_1 - \lambda_i} \leq \frac{\lambda_i}{2\lambda_1 - 1}$, for all $2 \leq i \leq d$. This gives the following bound on the norm.  
        \[
        \norm{\E_{\fakedist}[\ve_{-1}]}_1 = \sum_{i=2}^{d} \E[e_i] \leq \frac{\sum_{i=2}^{d} \lambda_i}{2 \lambda_1 - 1} = \frac{1 - \lambda_1}{2 \lambda_1 - 1}.
        \]
        On the other hand, 
        \begin{align*}
            \Pr[e_1 < 0] \E_{\ve \sim \fakedist}[e_2 + \cdots + e_d \mid e_1 < 0]
            &= \Pr[e_2 + \cdots + e_d > n] \E_{\ve \sim \fakedist}[e_2 + \cdots + e_d \mid e_2 + \cdots + e_d > n] \\
            &= \sum_{j=n+1}^{\infty} j \cdot \Pr[e_2 + \cdots + e_d = j] \\
            &\leq \sum_{j=n+1}^{\infty} j \parens*{\frac{1- \lambda_1}{\lambda_1}}^j \\
            &= \parens*{\frac{1-\lambda_1}{\lambda_1}}^{n+1} \frac{\lambda_1}{(2 \lambda_1 - 1)^2} (n(2 \lambda_1 - 1) + \lambda_1) \\
            &= \Delta \parens*{n + \frac{\lambda_1}{2 \lambda_1 - 1}}
        \end{align*}
        \begin{align*}
            \norm{\E_{\truedist}[\ve] - \E_{\fakedist}[\ve]}_{1} 
            &=\frac{2p}{1-p} \parens*{\norm{\E_{\fakedist}[\ve_{-1}]}_1 + \norm{\E_{\fakedist}[\ve_{-1} \mid e_1 < 0]}_1} \\
            &\leq \frac{2\Delta}{1-\Delta} \parens*{\frac{1 - \lambda_1}{2 \lambda_1 - 1} + n + \frac{\lambda_1}{2 \lambda_1 - 1}} \\
            &= \frac{2 \Delta}{1 - \Delta} \parens*{n + \frac{1}{2 \lambda_1 - 1}}
        \end{align*}
    \end{proof}

    Finally, we bound the variance by going through the covariance matrices. 
    \covariance*
    \begin{proof}
        An elegant way to write the covariance matrix is  
        \[
        \Sigma_{\fakedist} = \E_{\substack{\ve \samp \fakedist \\ \ve' \samp \fakedist}}[(\ve - \ve')(\ve - \ve')^{\top}].
        \]
        Each $(\ve - \ve')(\ve - \ve')^{\top}$ is positive semidefinite, and therefore so is $\Sigma_{\fakedist}$. Now split the expectation based on whether $e_1$ and $e_1'$ are nonnegative.
        \begin{align*}
            \Sigma_{\fakedist} &= \E_{\substack{\ve \samp \fakedist \\ \ve' \samp \fakedist}}[(\ve - \ve')(\ve - \ve')^{\top} \mid e_1 \geq 0 \wedge e_1' \geq 0] \Pr_{\ve \samp \fakedist}[e_1 \geq 0]^2 + \phantom{.} \\
            &\phantom{=.} \E_{\substack{\ve \samp \fakedist \\ \ve' \samp \fakedist}}[(\ve - \ve')(\ve - \ve')^{\top} \mid e_1 < 0 \vee e_1' < 0] \parens*{1 - \Pr_{\ve \samp \fakedist}[e_1 \geq 0]^2}
        \end{align*}
        Recall that $\truedist$ is $\fakedist$ conditioned on $e_1 \geq 0$ (Lemma~\ref{lem:true_is_conditional_of_fake}), so that first conditional expectation is actually $\Sigma_{\truedist}$.
        \[
            \E_{\substack{\ve \samp \fakedist \\ \ve' \samp \fakedist}}[(\ve - \ve')(\ve - \ve')^{\top} \mid e_1 \geq 0 \wedge e_1' \geq 0] = \E_{\substack{\ve \samp \truedist \\ \ve' \samp \truedist}}[(\ve - \ve')(\ve - \ve')^{\top}] = \Sigma_{\truedist} 
        \]
        It follows that 
        \[
        \Sigma_{\fakedist} \succeq \Sigma_{\truedist} \Pr_{\ve \samp \fakedist}[e_1 \geq 0]^2 \succeq \Sigma_{\truedist} (1 - \Delta)^{2}.
        \]    
    \end{proof}

\section{Optimal Parameter Choice}
\label{sec:optimal_parameter_choice}
 
 We now prove Theorem~\ref{thm:optimal_parameter_choice}.
 \optimalchoice*
 For convenience, we duplicate~\cref{table:choice} below.

\begin{table}[h]
    \centering
    \scalebox{1}{
    \begin{tabular}{c|c|c|c}
    $\eta$ & $\bigo{1/s^{*}}$ & $\Omega(1/s^{*}) \cap \bigo{\sqrt{\epsilon}}$ & $\Omega(\sqrt{\epsilon})$ \\
    \hline
    $k$ & $1$ & $1$ & $\bigo{\frac{\initeta}{\sqrt{\epsilon}}}$ \\
    $n$ & $\bigo{s^*}$ & $\bigo{\frac{1}{\initeta}}$ & $\bigo{\frac{1}{\sqrt{\epsilon}}}$ \\
    $b$ & $1$ & $\bigo{\frac{B\initeta^2+\initeta}{\epsilon^2}}$ & $\bigo{\frac{B}{\epsilon}+\frac{1}{\epsilon^{3/2}}}$ \\
    \hline
    $s$ & $\bigo{s^*}$ & $\bigo{\frac{B \eta + 1}{\epsilon^2}}$ & $\bigo{\frac{B\eta}{\epsilon^2} + \frac{\eta}{\epsilon^{5/2}}}$
    \end{tabular}
    }
    \caption{Choice of parameters $k$, $n$, $b$ for the three regimes of $\eta$. Recall $s^* := \frac{\sqrt{B}}{\epsilon} + \frac{1}{\epsilon^2}$.}
    \label{table:choice_copy}
\end{table}

\begin{proof}

    The optimal choice of parameters $k$, $n$, $b$ is essentially a mathematical program. The objective is $s$, the expected number of samples of $\initrho$. Figure~\ref{fig:combined_subprocedures} makes it clear that $k$, $n$, $b$ multiply, and since that only counts the \emph{successful} measurements, we furthermore multiply by a factor $1/\Pr[\success]$ in expectation. We recall that $\Pr[\success]$ falls off exponentially with $n$ (cf.\  \cref{eq:meas_fail_prob}). It is therefore advisable to keep $\Pr[\success]$ above a constant, which is achieved by setting $n = \bigo{1/\eta'} = \bigo{k / \eta}$. This gives us both our objective value, $s = knb/\Pr[\success] = \Theta(knb)$, and first constraint, Equation \eqref{eq:success_constraint}.

    The remaining constraints come from correctness, i.e., the requirement to output an estimator with error at most $\epsilon$. The mean squared error of our estimate is split between the bias squared and the variance. From \cref{eq:purified_deviation} and \cref{eq:bias_norm_operator}, we deduce that the bias $\beta$ is at most $\bigo{\frac{\initeta}{kn}}$. From \cref{thm:robustness_of_joint_measurement}, the total variance of our estimator satisfies:
    $
        \mathrm{Var}(\Tr(O \hat\phi(b)))=\mathcal O(\frac{B}{n^2 b}+\frac{1}{nb}).
    $
    To ensure the mean squared error is at most $\bigo{\epsilon^2}$, the bias should be $\bigo{\epsilon}$, and each term of the variance at most $\bigo{\epsilon^2}$ (or standard deviation at most $\bigo{\epsilon}$). The variance terms translate directly to constraints \eqref{eq:variance_constraint_1}, and \eqref{eq:variance_constraint_2}, and the bias condition gives \eqref{eq:bias_constraint}. Last, we require $k, n, b$ to be at least $1$. 

\begin{align}
    \text{minimize } &&& knb \notag \\
    \text{subject to} && \frac{\eta}{nk} &= \bigo{\epsilon} \quad \quad &\text{(bias condition)} \label{eq:bias_constraint} \\
    && \frac{B}{n^2 b} &= \bigo{\epsilon^2} \quad \quad &\text{(variance condition 1)} \label{eq:variance_constraint_1} \\
    && \frac{1}{nb}&= \bigo{\epsilon^2} \quad \quad &\text{(variance condition 2)} \label{eq:variance_constraint_2} \\
    && n &= \bigo{k/\initeta} \quad \quad &\text{(success condition)} \label{eq:success_constraint} \\
    && k, n, b &\geq 1. & \text{(positivity condition)} \label{eq:positivity_constraint}
\end{align}
It remains to optimize this program for arbitrary $B$, $\epsilon$, $\eta$. 

Table~\ref{table:choice} gives optimal values for $k$, $n$, and $b$ in each of the three regimes. It is a calculation to see that these solutions are feasible and achieve the claimed sample complexities. On the other hand, optimality can be certified by the following products of constraints. 
\begin{align}
    \sqrt{\eqref{eq:variance_constraint_1}} \times (1 \leq \sqrt{b}) \times (1 \leq k) &\implies & knb &= \Omega(\sqrt{B}/\epsilon) \label{eq:dual1} \\
    \eqref{eq:variance_constraint_2} \times (1 \leq k) &\implies & knb &= \Omega(1/\epsilon^2) \label{eq:dual2} \\
    \eqref{eq:variance_constraint_1} \times \eqref{eq:success_constraint} &\implies & knb &= \Omega(B \eta / \epsilon^2) \label{eq:dual3} \\
    \sqrt{\eqref{eq:bias_constraint}} \times \eqref{eq:variance_constraint_2} \times \sqrt{\eqref{eq:success_constraint}} &\implies & knb &= \Omega(\eta / \epsilon^{5/2}) \label{eq:dual4}
\end{align}
For example, if we multiply $1/nb = \mathcal O(\epsilon^2)$ and $1 \leq k$ we get $1/nb = \mathcal O(k\epsilon^2)$, which we can rearrange to $knb = \Omega(1/\epsilon^2)$, \eqref{eq:dual2}. 
It is clear from these equations that the complexities of the three regimes---$\mathcal O(\sqrt{B}/\epsilon + 1/\epsilon^2)$, $\mathcal O((B\eta + 1)/\epsilon^2)$, and $\mathcal O(B \eta/\epsilon^2 + \eta/\epsilon^{5/2})$---arise from \eqref{eq:dual1} + \eqref{eq:dual2}, \eqref{eq:dual2} + \eqref{eq:dual3}, and \eqref{eq:dual3} + \eqref{eq:dual4} respectively. Likewise, thresholds between regimes are given by the crossover points of these inequalities, i.e., $\sqrt{B}/\epsilon \approx B\eta/\epsilon^2$ implies $\eta \approx \epsilon/\sqrt{B} = \Theta(1/s^{*})$ and $1/\epsilon^2 \approx \eta/\epsilon^{5/2}$ gives $\eta \approx \sqrt{\epsilon}$. 
\end{proof}

\section{Eigenvalue estimation}
\label{app:estimateeta}

In this appendix, we prove the following theorem, as a component of estimating $\eta$ from samples of $\rho$. 
\estimateeta*

The algorithm is simple: it performs/samples Bernoulli trials until it has seen $r$ failures total. Then it outputs $r/T$ as an estimate for $p$, where $T$ is the total number of \emph{trials}.
\begin{lemma}
    The number of trials $T$ is bounded above and below with high probability, i.e., 
    \[
    \Pr[T = \Theta(r/p)] \geq 1 - \exp(-\Theta(r)).
    \]
    Concretely, we have, e.g., 
    \begin{align*}
        \Pr[\ln(2) \tfrac{r}{p} \leq T \leq \ln(4) \tfrac{r}{p}] &\geq 1 - 2 \cdot (\tfrac{1}{2} e \ln 2)^{r} \geq 1 - 2 \cdot 0.9421^{r}.
    \end{align*}
\end{lemma}
\begin{proof}
    First, we argue that for any $n \geq r$, if $Y \sim \mathrm{Bin}(n, p)$ then $\Pr[T \leq n] = \Pr[Y \geq r]$. To see this, imagine an infinite sequence of Bernoulli trials. The binomial distribution counts the number of failures, $Y$, in the first $n$ trials, whereas the algorithm scans down the list to the $r$th failure at some position $T$. It is clear that the $r$th failure happens at or before position $n$ ($T \leq n$) if and only if there are $r$ or more failures among the first $n$ trials ($Y \geq r$). 

    Since the mean of the binomial is $\E[Y] = np$, a multiplicative Chernoff bound gives 
    \[
    \Pr[Y \geq (1+x)np] \leq \parens*{ \frac{e^{x}}{(1+x)^{1+x}} }^{np}.
    \]
    Set $n = \frac{r}{(1+x)p}$ so that $r = (1+x)np$, and this becomes 
    \[
    \Pr \bracks*{T \leq \frac{r}{(1+x)p}} = \Pr[Y \geq r] \leq \parens*{\frac{e^{x/(1+x)}}{1+x}}^{r}.
    \]
    For instance, at $x^{*} = \frac{1}{\ln 2} - 1$, we have $\Pr[T \leq \tfrac{r}{p} \ln 2]\leq (\tfrac{1}{2} e \ln 2)^{r} \approx 0.9421^{r}$. 

    On the other side, $\Pr[T > n] = \Pr[Y < r] \leq \Pr[Y \leq r]$. The other side of the Chernoff bound gives 
    \[
    \Pr[Y \leq (1-x)np] \leq \parens*{ \frac{e^{-x}}{(1-x)^{(1-x)}}}^{np}.
    \]
    Setting $r = (1-x)np$, we translate this to 
    \[
    \Pr \bracks*{T \geq \frac{r}{(1-x)p}} \leq \parens*{\frac{e^{-x/(1-x)}}{1-x}}^{r}.
    \]
    At $x^{*} = 1 - \frac{1}{\ln 4}$ we get $\Pr[T \geq \tfrac{r}{p} \ln 4] \leq (\tfrac{1}{2} e \ln 2)^{r} \approx 0.9421^{r}$. Union bound over the two tail bounds finishes the result. 
\end{proof}

\end{document}